%% file: scr.tex
\documentclass[letterpaper,twocolumn,10pt]{article}
\usepackage{usenix2019_v3}

\usepackage{tikz}
\usepackage{amsmath}
\usepackage{subfig}
\usepackage{wrapfig}
\usepackage{filecontents}
\usepackage{caption}
\usepackage{xspace}
\usepackage{xcolor}
\usepackage{multirow}
\usepackage{graphicx}
\usepackage{algorithm} 
\usepackage{algpseudocode}
\usepackage{amssymb}
\usepackage{amsthm}

\input{defs}

\begin{document}
%-------------------------------------------------------------------------------

%don't want date printed
\date{}

% make title bold and 14 pt font (Latex
% default is non-bold, 16 pt)

\title{\Large \bf State-Compute Replication: Parallelizing High-Speed
  Stateful Packet Processing}

% \author{Submission \#349 (revision of \#1010 from fall'24)}
\author{
Qiongwen Xu\textsuperscript{*},
Sebastiano Miano\textsuperscript{\dag},
Xiangyu Gao\textsuperscript{\ddag},
Tao Wang\textsuperscript{\ddag},
Adithya Murugadass\textsuperscript{*},
Songyuan Zhang\textsuperscript{*},\\
Anirudh Sivaraman\textsuperscript{\ddag},
Gianni Antichi\textsuperscript{\dag,\S},
Srinivas Narayana\textsuperscript{*}
\\
\textsuperscript{*}Rutgers University,
\textsuperscript{\dag}Politecnico di Milano,
\textsuperscript{\ddag}New York University,
\textsuperscript{\S}Queen Mary University of London
}

\maketitle

\input{abstract}

\input{introduction}
\input{background}

\input{design}

\input{pipeline-reliable-sequencer}

\input{packet-loss-nondeterminism}

\input{evaluation}

\input{discussion}
\input{related}
\input{conclusion}

\bibliographystyle{plain}
\bibliography{scr.bib}

\clearpage
\appendix
\input{throughput-model}

\input{loss-recovery}
\input{running-example-programming}

\end{document}

%% file: defs.tex
\newcommand{\cut}[1]{}
\newcommand{\ie}{\emph{i.e.}\xspace}
\newcommand{\etc}{\emph{etc.}\xspace}
\newcommand{\eg}{\emph{e.g.}\xspace}

\newcommand{\Fig}[1]{Figure \ref{fig:#1}\xspace}
\newcommand{\Tab}[1]{Table \ref{tab:#1}\xspace}
\newcommand{\ngs}[1]{\textcolor{red}{#1}}
\newcommand{\SCR}{{SCR}\xspace}
\newcommand{\nop}[1]{}
\newcommand{\dispatch}{dispatch\xspace}

\newcommand{\Para}[1]{\vspace{4pt}\noindent\textbf{\textit{#1}}}

\def\compactify{\itemsep=0pt \topsep=0pt \partopsep=0pt \parsep=0pt \leftmargin=12pt}
\let\latexusecounter=\usecounter

\newenvironment{CompactEnumerate}
  {\def\usecounter{\compactify\latexusecounter}
   \begin{enumerate}}
  {\end{enumerate}\let\usecounter=\latexusecounter}

\newcommand{\srinivas}[1]{}
\newcommand{\qx}[1]{}

\newcommand{\ct}{\small \tt}
\newcommand{\Sec}[1]{\S\ref{sec:#1}\xspace}
\newcommand{\App}[1]{App.\ref{app:#1}\xspace}
\newcommand{\Alg}[1]{Alg.\ref{alg:#1}\xspace}

\newcommand{\revadd}[1]{{\textcolor{black}{#1}}}
\newcommand{\revaddlargebegin}{\color{black}}
\ifdefined\issupplement
  \renewcommand{\revadd}[1]{{\textcolor{blue}{#1}}}
  \renewcommand{\revaddlargebegin}{\color{blue}}
\fi
\newcommand{\forcameraready}[1]{}

\usepackage{listings}

%% file: abstract.tex
\begin{abstract}

With the slowdown of Moore's law, CPU-oriented packet processing in
software will be significantly outpaced by emerging line speeds of
network interface cards (NICs). Single-core packet-processing
throughput has saturated.

We consider the problem of high-speed packet processing with multiple
CPU cores. The key challenge is state---memory that multiple packets
must read and update. The prevailing method to scale throughput with
multiple cores involves state sharding, processing all packets that
update the same state, e.g., flow, at the same core. However, given
the skewed nature of realistic flow size distributions, this
method is untenable, since total throughput is limited by
single core performance.
  
This paper introduces {\em state-compute replication}, a principle to
scale the throughput of a single stateful flow across multiple cores
using replication. Our design leverages a {\em packet history
  sequencer} running on a NIC or top-of-the-rack switch to enable
multiple cores to update state without explicit synchronization. Our
experiments with realistic data center and wide-area Internet traces
shows that state-compute replication can scale total packet-processing
throughput linearly with cores, independent of
flow size distributions, across a range of realistic packet-processing
programs.

\end{abstract}

%% file: introduction.tex
\section{Introduction}

Designing software to handle high packet-processing loads is crucial
in networked systems. For example, software load balancers, CDN nodes,
DDoS mitigators, and many other middleboxes depend on it.  Yet, with
the slowdown of Moore's law, software packet processing has struggled
to keep up with line speeds of network interface cards (NICs), with
emerging speeds of 200 Gbit/s and
beyond~\cite{mellanox-connectx-7}. Consequently, there have been
significant efforts to speed up packet processing through better
network stack design, removing user-kernel crossings, running software
at lower layers of the stack (\eg NIC device driver), and designing
better host interconnects.

We consider the problem of scaling software packet processing by using
multiple cores on a server. The key challenge is that many
packet-processing applications are {\em stateful}, maintaining and
updating regions of memory across many packets. If multiple cores
contend to access the same memory regions, there is significant memory
contention and cache bouncing, resulting in poor performance. Hence,
the classic approach to multicore scaling is to process packets
touching distinct states, \ie flows, on different cores, hence
removing memory contention and synchronization.  For example, a load
balancer that maintains a separate backend server for each 5-tuple may
send all packets of a given 5-tuple to a fixed core, but process
different 5-tuples on different cores, hence scaling performance with
multiple cores. Many prior efforts have looked into optimizing such
sharding-oriented solutions~\cite{opennf-sigcomm14,
  split-merge-nsdi13, metron-nsdi18, rss++-conext19}, including the
recent application of automatic code parallelization
technology~\cite{automatic-parallelization-nsdi24}.

However, we believe that the existing
approaches to multi-core scaling have run their course
(\Sec{motivation-background}).  Realistic traffic workloads have
heavy-tailed flow size distributions and are highly skewed. Large
``elephant flows'' updating state on a single core will reduce total
throughput and inflate tail latencies for all packets, since they are
limited by the performance of a single CPU core. With emerging 200
Gbit/s---1 Tbit/s NICs, a single packet processing core may be too
slow to keep up even with a single elephant flow. Additionally, with
the growing scales of volumetric resource exhaustion attacks, packet
processors must gracefully handle attacks where adversaries force
packets into a single flow~\cite{maglev-nsdi16}.

This paper introduces a scaling principle, {\em state-compute
  replication (\SCR)}, that improves the software packet-processing
throughput for a {\em single, stateful flow} with additional cores,
while avoiding shared memory and
contention. \Fig{throughput-conntrack-singleflow} shows how \SCR
scales the throughput of a single TCP connection for a TCP connection
state tracker~\cite{conntrack} with more cores, when other scaling
techniques fail. A connection tracker may change its internal state
with each packet.

\begin{figure}
  %%\vspace{-5mm}
  \centering
  \includegraphics[width=0.30\textwidth]{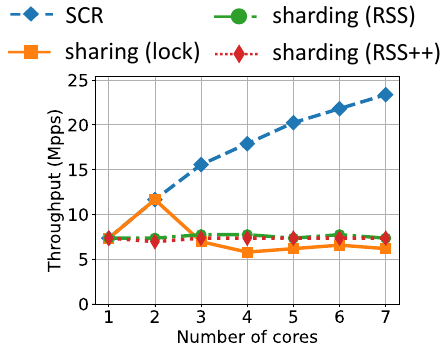}
  \vspace{-2mm}
  \caption{Scaling the throughput of a TCP connection state tracker
    for a {\em single TCP connection} across multiple cores. Sharing
    state across cores degrades performance beyond 2 cores due to
    contention. Sharding state (using RSS and
    RSS++~\cite{rss++-conext19}) cannot improve throughput beyond a
    single CPU core (\Sec{motivation-background}). In contrast,
    State-Compute Replication (\Sec{design}) provides linear scale-up
    in throughput with cores.}
  \vspace{-6mm}
  \label{fig:throughput-conntrack-singleflow}
\end{figure}

\SCR applies to any packet processing program that may be abstracted
as a deterministic finite state machine. Intuitively, as long as each
core can reconstruct the flow state by processing {\em all} the
packets of the flow, multiple cores can successfully process a single
stateful flow with zero cross-core synchronization. However, naively
getting every CPU core to process every packet cannot improve
throughput with more cores. We need to meet two additional
requirements. First, in systems where CPU usage is dominated by
per-packet work, we must preserve the total number of packets moving
through the system. Second, we must divide up {\em \dispatch}---the
CPU software labor of presenting the packets to the packet-processing
program---across multiple cores (\Sec{scr-principle}). This leads to
the scaling principle below (more formal version in
\Sec{scr-principle}):

\noindent
\textbf{State-Compute Replication (informal)}. In a system bound by
per-packet CPU work and software dispatch, {\em replicating}
the state and the computation of a stateful flow across cores
increases throughput {\em linearly} with cores, so long as we preserve
the total number of packets traversing the system.

To meet the requirements for high performance in \SCR, we design a
{\em packet history sequencer} (\Sec{operationalizing-scr}), an entity
which sees every packet and sprays packets across cores in a
round-robin fashion, while piggybacking the relevant information from
packets missed by a core on the next packet sent to that core. The
sequencer must maintain a small bounded history of packet headers
relevant to the packet-processing program.  A high-speed
packet-processing pipeline, running either on a programmable NIC or a
top-of-the-rack switch, may serve as a sequencer. We present sequencer
designs for two hardware platforms (\Sec{pipeline-sequencer}), a
Tofino switch, and a Verilog RTL module we synthesized into the
NetFPGA-PLUS reference pipeline.

We evaluated the scaling efficacy of \SCR using a suite of realistic
programs and traffic (\Sec{evaluation}).  \SCR is the only technique
we are aware of that monotonically scales the total processing
throughput with additional cores regardless of the skewness in the
arriving flow size distribution. However, no technique can continue
scaling indefinitely. In \Sec{scr-principle} and \Sec{evaluation}, we
demonstrate the limits of \SCR scaling, along with the significant
performance benefits to be enjoyed before hitting these
limits. Additionally, we show the resource usage for our sequencer's
RTL design, which meets timing at 250 MHz. We believe this design is
simple and cheap enough to be added as an on-chip accelerator
to a future NIC. We plan to release all software.
%% on today's hardware.
%% under realistic packet-processing programs and packet traffic
%% workloads.

%% file: background.tex
\section{Background and Motivation}
\label{sec:motivation-background}

\subsection{High Speed Packet Processing}
\label{sec:background-high-speed-packet-processing}

This paper considers packet-processing programs that must work at
heavy network loads with quick turnarounds for packets.  We are
specifically interested in applications implementing a ``hairpin''
traffic flow: receiving a packet, processing it as quickly as
possible, and transmitting the packet right out, at the highest
possible rate. Such programs exhibit compute and working set sizes
that are smaller in comparison to full-fledged applications (running
in user space) or transport protocols (TCP stacks) running at
endpoints.

Examples of the kinds of applications we consider include (i)
middlebox workloads (network functions), such as firewalls, connection
trackers, and intrusion detection/prevention systems; and (ii)
high-volume compute-light applications such as key-value stores,
telemetry systems, and stream-based logging systems, which process
many requests with a small amount of computation per
request~\cite{nanopu-osdi21}. A key
characteristic of such applications is their need for high
packet-processing rate (which is more important than byte-processing
rate) and the fact that their performance is primarily bottlenecked by
CPU usage~\cite{netmap-atc12, katran-facebook-talk,
  cloudflare-l4drop}.

%% TODO: we should actually evaluate our techniques on such
%% applications if possible

%% TODO: Wondering if we should include a note about memory
%% usage. These apps should also have a small memory
%% footprint. Otherwise, it is difficult to function at high speeds --
%% even a single cache miss will affect performance. Also, for our
%% technique when we replicate state, the total memory footprint
%% becomes much higher, so this is even more crucial.

The performance of such applications is mission-critical in many
production systems. Even small performance improvements matter in
Internet services with large server fleets and heavy traffic. As a
specific example, Meta's Katran layer-4 load balancer~\cite{katran}
and CloudFlare's DDoS protection solution~\cite{cloudflare-l4drop}
process every packet sent to those respective services, and must
withstand not only the high rate of request traffic destined to those
services but also large denial-of-service attacks. More generally, the
academic community has undertaken significant efforts for performance
optimization of network functions, including optimization of the
software frameworks~\cite{opennf-sigcomm14}, designing language-based
high-performance isolation~\cite{netbricks-osdi16}, and developing
custom hardware offload solutions~\cite{floem-osdi18}.

%% TODO: add citations for the last sentence above from the NF
%% literature

In this paper, we consider applications developed within high-speed
packet processing software frameworks.  Given the slowdown of Moore's
law and the end of Dennard scaling, the software packet-processing
performance of single CPU cores has saturated. Even expert developers
must work meticulously hard to improve per-core throughput by small
margins (like 10\%)~\cite{xdp-conext18, accelerate-with-af-xdp,
  kubernetes-load-balancing-with-xdp,
  netronome-smaller-programs-greater-performance}. The community has
pursued various efforts to improve performance, such as
re-architecting the software stack to make efficiency
gains~\cite{netmap-atc12, dpdk, snap-google-sosp19}, introducing
stateless hardware offloads working in conjunction with the software
stack~\cite{segmentation-offloads, checksum-offloads}, full-stack
hardware offloads~\cite{chelsio, tonic-nsdi20}, and kernel
extensions~\cite{kernel-bpf-documentation}.  This paper considers
software frameworks that modify the device driver to enable
programmable and custom functionality to be incorporated by developers
at high performance with minimal intervention from the existing kernel
software stack~\cite{xdp-conext18}.

Specifically, we study performance and evaluate our techniques in the
context of kernel extensions implemented using the eXpress Data Path
(XDP/eBPF~\cite{xdp-conext18}) within the Linux kernel. In
\Sec{design}, we will discuss how our observations and principles
apply more generally to other high-speed software packet-processing
frameworks, including those written with user-space libraries like
DPDK~\cite{dpdk}.

%% TODO: there needs to be some justification, or experiment, to show
%% why our principles will hold with DPDK. 

\subsection{Parallelizing Stateful Packet Processing}
\label{sec:background-parallelizing-stateful-packet-processing}

This paper considers the problem of parallelizing high-speed packet
processing programs across multiple cores. The key challenge is
handling {\em state}: memory that must be updated in a sequential
order upon processing each packet to produce a correct
result. Consider the example of the {connection
  tracker}~\cite{conntrack}, a program which identifies the TCP
connection state (\eg SYN sent, SYN/ACK sent, \etc) using packets
observed from both directions of a TCP connection. Each packet in the
connection may modify the internal connection state maintained by the
program. There are two main techniques used to parallelize such
programs across cores.

\Para{Shared state parallelism.} One could conceive a parallel
implementation that (arbitrarily) splits packets across multiple
cores, with explicit synchronization or retry logic guarding access to the shared
memory, \ie the TCP connection state, to ensure a correct result.

Shared-state parallelism works well when the contention to shared
memory is low. Specifically, shared-memory scaling could work well
when (i) packets of a single flow arrive slowly enough, \eg if there
are a large number of connections with a roughly-equal packet arrival
rate, or (ii) when there are efficient implementations available for
synchronization or atomic updates in
software~\cite{conntrack-hpsr21} or
hardware~\cite{gcc-builtin-atomics,
  gcc-legacy-builtin-atomics}. However, neither of these conditions
are generally applicable. Many flow size distributions encountered in
practical networks are highly skewed~\cite{microsoft-network-sigcomm10,
  internet-flow-rates-sigcomm02} or exhibit highly bursty
behavior~\cite{facebook-datacenter-study-sigcomm15}, resulting in
significant memory contention if packets from the heavier flows are
spread across cores. Further, the state update operation in many
programs, including the TCP connection tracker, are too complex to be
implemented directly on atomic hardware, since the latter only
supports individual arithmetic and logic operations (like
fetch-add-write). Our evaluation results (\Sec{evaluation}) show that
the performance of shared-state multicore processing plummets with
more cores under realistic flow size distributions.

\Para{Sharded (shared-nothing) parallelism.} Today, the predominant
technique to scale stateful packet processing across multiple cores is
to process packets that update the same memory at the same core,
sharding the overall state of the program across cores. Such sharding
is usually achieved through techniques like Receive Side Scaling
(RSS~\cite{rss}), available on modern NICs, to direct packets from the
same flow to the same core, and using shared-nothing data structures
on each core.

However, sharding suffers from a few disadvantages. First, it is not
always possible to avoid coordination through sharding.  There may be
parts of the program state that are shared across all packets, such as
a list of free external ports in a Network Address Translation (NAT)
application.  On the practical side, RSS implementations on today's
NICs partition packets across cores using a limited number of
combinations of packet header fields. For example, a NIC may be
configured to steer packets with a given source and destination IP
address to a fixed core. However, the granularity at which the
application wants to shard its state---for example, a key-value cache
may seek to shard state by the key requested in the payload---could be
infeasible to implement with the packet headers that are usable by RSS
at the NIC~\cite{automatic-parallelization-nsdi24}.

Second, sharding state may create load imbalance across cores if some
flows are heavier or more bursty than others, creating hotspots on
single CPU cores.  Skewed flow size
distributions~\cite{microsoft-network-sigcomm10}, bursty flow
transmission patterns~\cite{facebook-datacenter-study-sigcomm15}, and
denial of service attacks~\cite{maglev-nsdi16} create conditions ripe
for such imbalance. The research community has investigated solutions
to balance the packet processing load by migrating flow shards across
CPU cores~\cite{rss++-conext19, automatic-parallelization-nsdi24}.
However, the efficacy of re-balancing
is limited by the granularity at which flows can be migrated across
cores. As we show in our evaluation (\Sec{evaluation}), the throughput
of the heaviest and most bursty flows is still limited by a single CPU
core, which in turn limits the total achieved throughput. Another
alternative is to evenly spray incoming packets across
cores~\cite{spraying-in-middleboxes-hotnets18, rpcvalet-asplos19},
assuming that only a small number of packets in each flow need to
update the flow's state. If a core receives a packet that must update
the flow state, the packet is redirected to a designated core that
performs all writes to the state for that flow.  However, the
assumption that state is mostly read-only is not universal, \eg a TCP
connection tracker may update state potentially on every
packet. Further, packet reordering at the designated core can lead to
incorrect results~\cite{rss++-conext19}.

\subsection{Goals}
\label{sec:goals}
\label{sec:motivation}

Given the drawbacks of existing approaches for multi-core scaling
discussed above
(\Sec{background-parallelizing-stateful-packet-processing}), we seek a
scaling technique that achieves the following goals:

\begin{CompactEnumerate}
\item {\em Generic stateful programming.} The technique must produce
  correct results for general stateful updates, eschewing solutions
  that only work for ``simple'' updates (\ie fitting hardware
  atomics) %% , or only programs compatible with NIC
  %% configuration,
  or only update state for a small number of packets per flow.
\item {\em Skew independence.} The scaling technique should improve
  performance independent of the incoming flow size distribution or
  how the flows access the state in the program.
\item {\em Monotonic performance gain.} Performance should
  improve, not degrade or collapse, with additional cores.
\end{CompactEnumerate}

%% file: design.tex
\section{State-Compute Replication (\SCR)}
\label{sec:design}

In \Sec{scr-principle}, we present scaling principles for multi-core
stateful packet processing, to meet the goals in
\Sec{goals}. In \Sec{operationalizing-scr} through
\Sec{packet-loss-nondeterminism}, we show how to operationalize these
principles.

\subsection{Scaling Principles}
\label{sec:scr-principle}

To simplify the discussion, suppose the packet-processing program is
{\em deterministic}, \ie in every execution, it produces the same
output state and packet given a fixed input state and packet (we relax
this assumption in \Sec{packet-loss-nondeterminism}).

\Para{Principle \#1 (Replication for correctness).} Sending every
packet reliably to every core, and {\em replicating} the
state and computation on every core, produces the correct output state
and packet on every core {\em with no explicit cross-core
  synchronization,} regardless of how the state is accessed by packets.

This principle asks us to treat each core as one replica of a
replicated state machine running the deterministic packet-processing
program. Each core processes packets in the same order, without
missing any packets. With each incoming packet, each core updates a
private copy of its state, which is equal to the private copy on every
other core. There is no need to synchronize explicitly. Further,
replication provides the benefit that the workload across cores is
even regardless of how the state is accessed by packets, \ie
skew-independent.

One way to apply this principle naively is to broadcast every packet
received {\em externally} on the machine to every core: with $k$
cores, for each external packet, the system will process $k$ {\em
  internal} packets, due to $k$-fold packet duplication. However,
artificially increasing the number of packets processed by the system
will significantly hurt performance.
%
%% \ngs{Insert picture with internal and external rates after principle 1
%%   and then after principle 2.}
%
In CPU-bound packet processing, smaller packets typically require the
same computation that larger packets do. That is, the total amount of
work performed by the system is proportional to the
packets-per-second offered, rather than the
bits-per-second~\cite{netmap-atc12, xdp-conext18}.

So how should one use replication for multi-core scaling?
Understanding the dominant components of the per-packet CPU work in
high-speed packet-processing frameworks offers insight.  There are two
parts to the CPU processing for each packet after the packet reaches
the core where it will be ultimately processed: (i) {\em dispatch},
the CPU/software labor of presenting the packet to the user-developed
packet-processing program, and signaling the packet(s) emitted by the
program for transmission by the NIC; and (ii) the {\em program
  computation} running within the user-developed program
itself. Dispatch often dominates the per-packet CPU
work~\cite{xdp-conext18}.

%% Many prior efforts use the term dispatch to refer to just the NIC
%% or a CPU thread moving packets to a CPU core; however, we use the
%% term to refer exclusively to the software work done on the final
%% CPU.

While these observations are known in the context of high-speed packet
processing, we also benchmarked a simple application on our own test
machine to validate them. Consider
\Fig{nature-of-packet-processing-cpu-work}, where we show the
throughput (packets/second (a), bits/second (b)) and latency (c) of a
simple packet forwarder written in the XDP framework running on a
single CPU core. Our testbed setup is described in much more detail
later (\Sec{evaluation}), but we briefly note here that our
device-under-test is an Intel Ice Lake CPU configured to run at a
fixed frequency of 3.6 GHz and attached to a 100 Gbit/s NIC. At each
packet size below 1024 bytes, the CPU core is fully utilized (at 1024
bytes the bottleneck is the NIC bandwidth). The achieved
packets/second is stable across all packet sizes which are CPU bound
(see the {\em 2 RXQ} curve). With a processing latency of roughly 14
nanoseconds at all packet sizes (measured only for the XDP forwarding
program), back-to-back packet processing should ``ideally'' process
$(14 * 10^{-9})^{-1}$ $\approx$ 71 million packets/second. However, the achieved
best-case throughput ($\approx$ 14 million packets/second) is much
smaller---implying that significant CPU work is expended in presenting
the input packets to and extracting the output packets from the
forwarder and setting up the NIC to transmit and receive those
packets. This is not merely a feature of the framework we used (XDP);
the DPDK framework has similar dispatch
characteristics~\cite{xdp-conext18}.

Our key insight is that it is possible to replicate program
computation without replicating dispatch, enabling multi-core
performance scaling. This leads us to the next principle:
%% Hence, despite replication, it is more important for CPU cores to
%% divide up the dispatch of the packet stream than the program
%% computations. This leads us to the following principle:

\begin{figure}
  \vspace{-5mm}
  \centering
  \subfloat[Packets/second]{
    \begin{minipage}[t]{0.15\textwidth}
    \includegraphics[width=\textwidth]{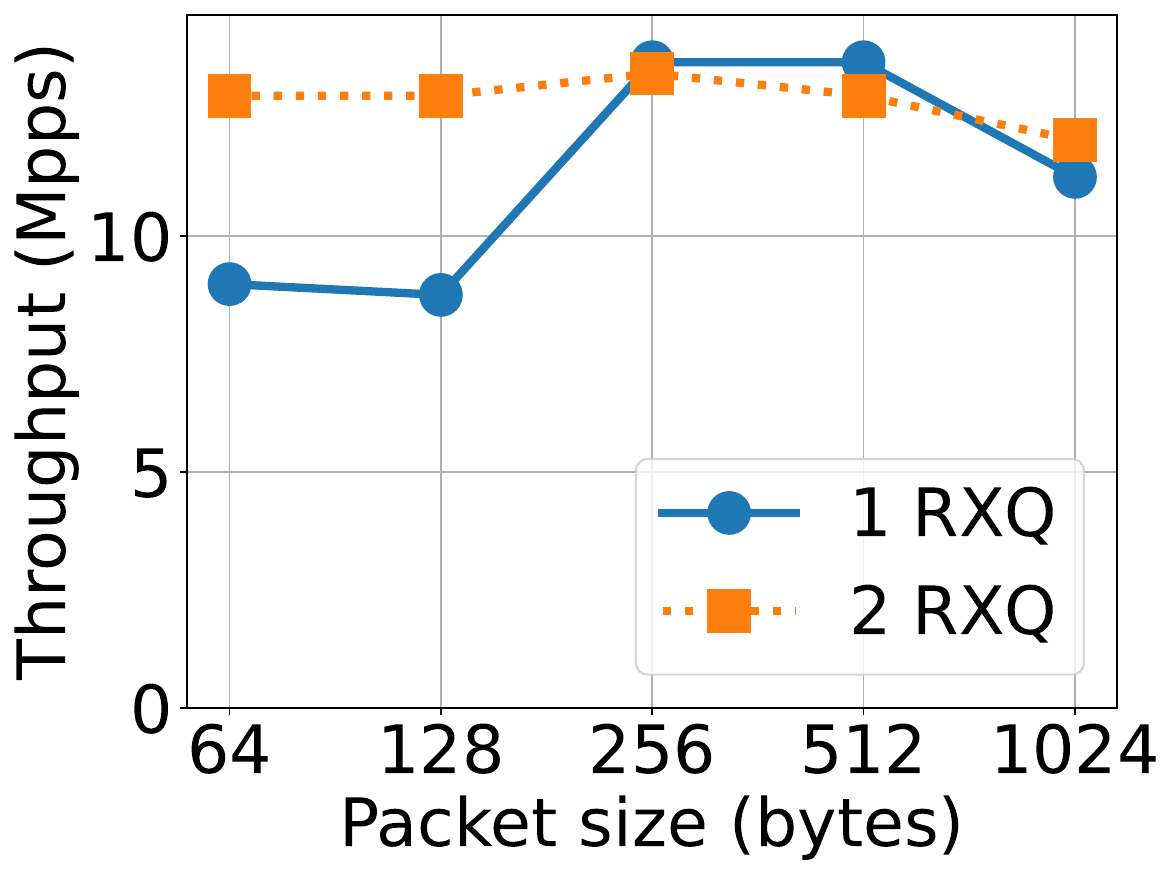}      
    \label{fig:xdp-fwd-rx-rate-pps}
    \vspace{-4mm}
    \end{minipage}} 
    % \hspace{0.01\textwidth}%
  \subfloat[Bits/second]{
    \begin{minipage}[t]{0.15\textwidth}
    \includegraphics[width=\textwidth]{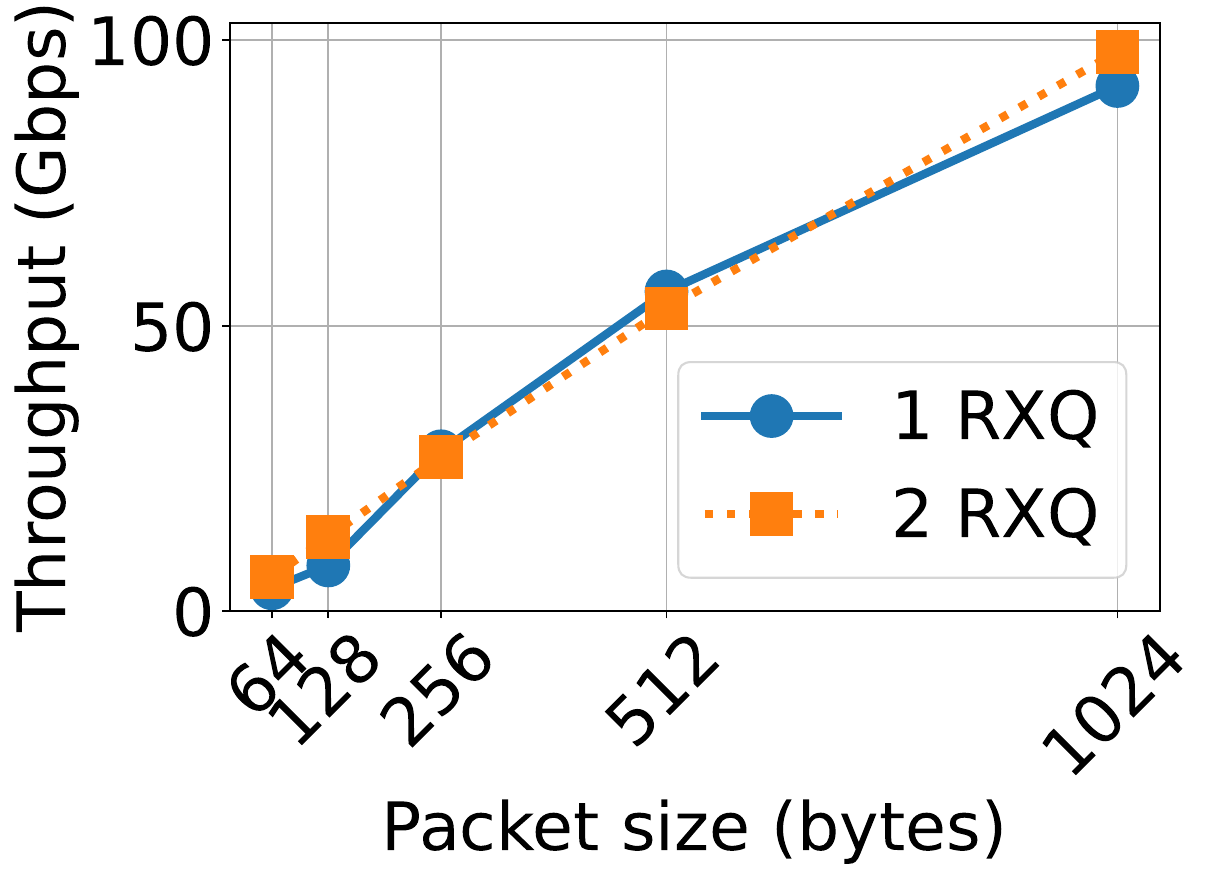}
    \label{fig:xdp-fwd-rx-rate-bps}
    \vspace{-4mm}
  \end{minipage}}
  \subfloat[Latency (ns)]{
    \begin{minipage}[t]{0.15\textwidth}
    \includegraphics[width=\textwidth]{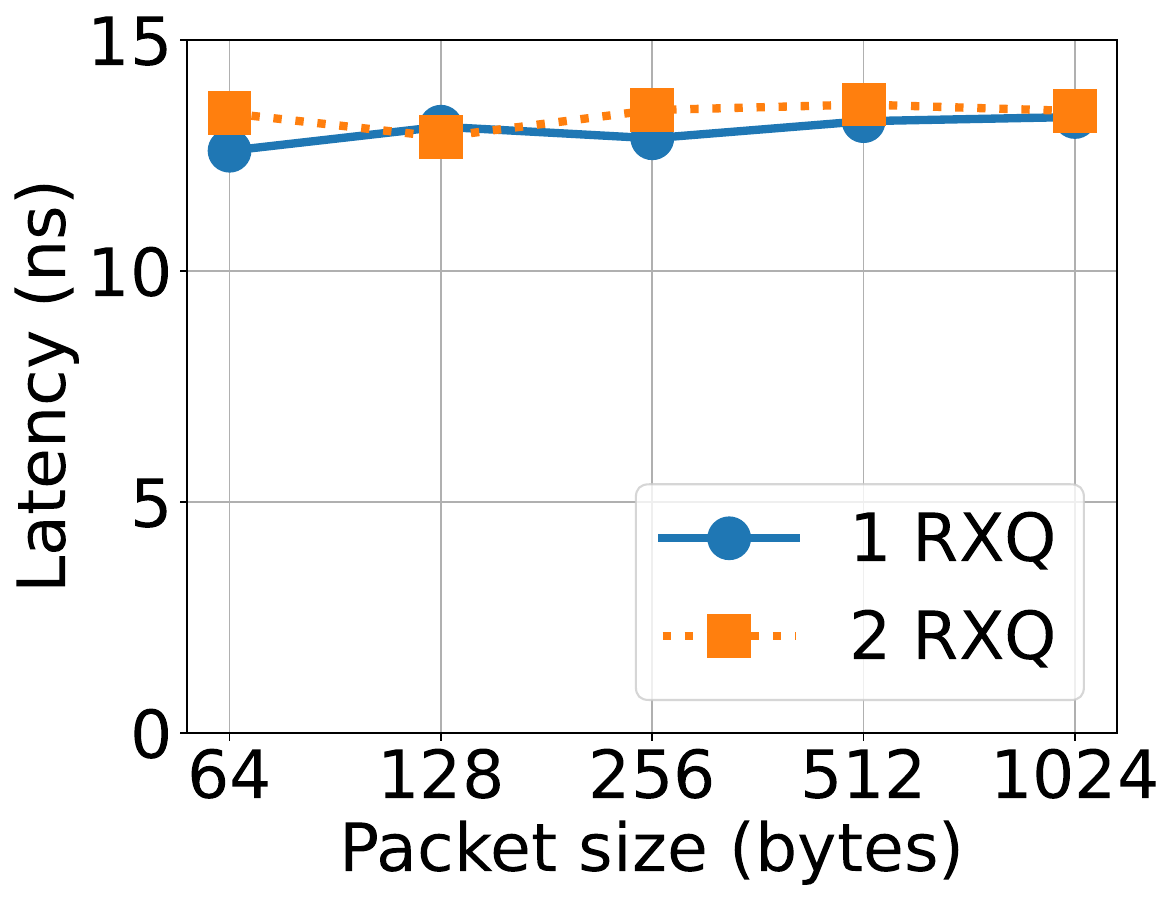}
    \label{fig:xdp-fwd-latency}
    \vspace{-4mm}
  \end{minipage}}
    %%\vspace{-2.5mm}
    \caption{The nature of CPU work in high-speed packet processing:
      Consider the throughput of a simple packet forwarding
      application (packets/second (a), bits/second (b)) running on a
      single CPU core clocked at 3.6 GHz, as the size of the incoming
      packets varies.  The average latency to execute the XDP program
      is also shown in nanoseconds (c). CPU usage is tied to the
      number of packets (not bits) processed per second. Further,
      significant time elapses in {\em dispatch:}
      CPU work to present the input packet to and retrieve the
      output packet from the program computation.}
    \label{fig:nature-of-packet-processing-cpu-work}
    \vspace{-5mm}
\end{figure}

\Para{Principle \#2 (State-Compute Replication).}  Piggybacking a {\em
  bounded recent packet history} on each packet sent to a core allows
us to use replication (\#1) while equalizing the external and internal
packets-per-second CPU work in the system.

Principle \#2 states that replication (principle \#1) is possible
without increasing the total internal packet rate or per-packet CPU
work done by the system.  Suppose it is possible to spray the
incoming packets across cores in a round-robin fashion. If there are
$k$ cores, each core receives every $k^{th}$ packet.  Then:

\begin{CompactEnumerate}
  \item It is unnecessary for each core to have the most
    up-to-the-packet private state at all times. For correctness, it
    is sufficient if a core has a ``fresh enough'' state to make a
    decision on the current packet that it is processing.
  \item With each new packet, suppose the core that receives it also
    sees (as metadata on that packet) all the $k-1$ packets from the
    last time it processed a packet, \ie\ {\em a recent packet
      history}. The core can simply ``catch up'' on the computations
    required to obtain the most up-to-the-packet value for its private
    state.
  \item If packets are sprayed round-robin across cores, the number of
    historic packets needed to ensure that the most updated state is
    available to process the current packet is equal to the
    number of cores. Further, just those packet bits required to
    update the shared state are necessary from each historic packet,
    allowing us to pack multiple packets' worth of history into a
    single packet received at a core.
\end{CompactEnumerate}

As a simple model, suppose a system has $k$ cores, and each core can
dispatch a single packet in $d$ cycles and runs a packet-processing
program that computes over a single packet in $c$ cycles. For each
piggybacked packet, the total processing time is $d + (k \times
c)$. When dispatch time dominates compute time, $d \gg c$. With $k$
cores, the total rate at which externally-arriving packets can be
processed is $k \times \frac{1}{d + (k \times c)} \approx k /
d$. Hence, it is possible to scale the packet-processing rate linearly
with the number of cores $k$. In \App{throughput-model}, we show that
a model like this indeed accurately predicts the empirical throughput
achieved with a given number of cores.

Intuitively, doing some extra ``lightweight'' program computation per
packet enables scaling the ``heavyweight'' dispatch computation with
more cores, while maintaining correctness.
%% this principle
%% trades off significant scaling in packet dispatch using more cores
%% with a little extra program computation, while maintaining
%% correctness.

\Para{Principle \#3 (Scaling limits).} Principle \#2 provides a linear
scale-up in the packets-per-second throughput with more cores, so long
as dispatch dominates the per-packet work.

%% The system's achievable packet rate will not scale linearly if
%% dispatch no longer dominates the per-packet work. 
The scaling benefits of principle \#2 taper off beyond a point. Dispatch can
be overtaken as the primary contributor to per-packet CPU work, for
example, when (i) the compute time $k \times c$ for each piggybacked
packet becomes sizable; (ii) the per-packet compute time $c$ itself
increases due to overheads in the system, \eg larger memory access
time when a core's copy of the state spills into a larger memory; or
(iii) other components such as the NIC or PCIe become the bottleneck
rather than the CPU. When this happens, the approximation in our
simple linear model ($d \gg c$) no longer holds, and the system's
packet rate no longer scales linearly with cores.

\subsection{Operationalizing \SCR}
\label{sec:operationalizing-scr}

Operationalizing the scaling principles discussed above
(\Sec{scr-principle}) conceptually requires two pieces.

\Para{A reliable packet history sequencer (\Sec{pipeline-sequencer}).}
We require an additional entity in the system, which we call a {\em
  sequencer,} to (i) steer packets across cores in round-robin
fashion, (ii) maintain the most recent packet history across all
packets arriving at the machine, and (iii) piggyback the 
history on each packet sent to the cores. After the packet is
processed by a CPU core, its piggybacked history can be stripped off
on the return path either at the core itself or at the sequencer.
%% The act of stripping off
%% the piggybacked history from the packet after it is processed by the
%% program can be implemented either at the CPU core or the
%% sequencer. 
\revadd{The size of the packet history depends only on the
  number of cores and metadata size (\Sec{scr-principle}) and
  is independent of the number of active flows.}

The NIC hardware or the top-of-the-rack switch are natural points to
introduce the sequencer functionality, since they observe all the
packets flowing into and out of the machine. Today's existing
fixed-function NICs do not implement the functionality necessary to
construct and piggyback a reliable packet history. However, we have
identified two possible instantiations that could, in the near future,
achieve this: (i) emerging NICs, \eg with programmable
pipelines~\cite{pensando, bluefield, intel-ipu}, implementing the full
functionality of the reliable sequencer; or (ii) a combination of a
NIC implementing round-robin packet steering~\cite{rss,
  intel-flow-director} and a programmable top-of-the-rack switch
pipeline~\cite{rmt-sigcomm13, tofino, trident} for maintaining and
piggybacking the packet history. We believe that either of these
instantiations may be realistic and applicable given the context: for
example, high-speed programmable NICs are already common in some large
production systems~\cite{firestone2018azure}, as are programmable
switch pipelines~\cite{sailfish-sigcomm21}.  We will show two possible
hardware designs in \Sec{pipeline-sequencer}. Hereafter, for brevity,
we refer to both of these designs as simply sequencers.

\Para{An SCR-aware packet-processing program.}  The packet-processing
program must be developed to replicate the program state and keep
private copies per core. Further, the program must process the packet
history first before computing over the current packet. We discuss a
running example that demonstrates how to transform a single-threaded
program to its \SCR-aware variant in \App{scr-programming}. We believe
that these program transformations can be automated in the future.

\input{example-figures}

\Para{An example showing scaling principles in action.} Consider
\Fig{overview}, where a sequencer and three cores are used to run a
packet-processing program. As shown in \Fig{pipeline}, the
sequencer sprays packets (\ie $p_i, p_{i+1}, \ldots$) in a round-robin
fashion across $k=3$ cores (\ie $core_1, core_2, core_3$). Further,
the sequencer stores the recent packet history consisting of the packet
fields from the last $k$ packets which are relevant to evolving the
flow state.
We denote the relevant part of a packet $p_i$ by
$f(p_i)$. For example, in a TCP connection tracking program, this includes
the TCP 4-tuple, the TCP flags, and sequence and ACK numbers.
%% Prior work on circular buffers implemented on
%% switches~\cite{packet-histories-nsdi21} has shown that it is feasible
%% to update packet histories at line rate on pipelines. 
Note that this
packet history is updated only by the sequencer and is never written to
by the cores.  In the example in \Fig{pipeline}, the packet history
supplied to $core_1$ processing packet $p_i$ is $f(p_{i-2}),
f(p_{i-1})$. As shown in \Fig{state-update}, each core updates its
local private state, first fast-forwarding the state by running the
program through the packet history $f(p_{i-2}), f(p_{i-1})$, and then
processing the packet $p_i$ sprayed to it.
%% As shown in Fig. \ref{fig:overview}(b), before
%% $core_1$ processes $p_i$, it will first process $f(p_{i-1}),
%% f(p_{i-2})$.  Since each core misses at most $n-1$ packets in
%% round-robin spraying, it can ``fast-forward'' its application state to
%% the most updated values using a packet history with a known, bounded
%% size of at most $n-1$.  This solution will always produce {\em
%%   correct} state and packet-level outcomes on each core.

If the packet-processing program is deterministic
(\Sec{scr-principle}), an \SCR-aware program is guaranteed to produce
the correct output state and packet if every CPU core is guaranteed to
receive the packets sent to it by the sequencer. We discuss how to
handle packet loss and nondeterminism in \Sec{packet-loss-nondeterminism}.

%% file: example-figures.tex
\begin{figure}
  \centering
  \subfloat[The sequencer stores relevant fields from the packet
    history, and piggybacks the history on packets sprayed
    round-robin across cores.]{
    \begin{minipage}[t]{0.45\textwidth}
      \begin{center}
        \includegraphics[width=0.80\textwidth]{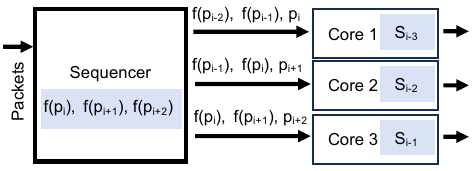}
    \label{fig:pipeline}
    \vspace{-3mm}
      \end{center}
    \end{minipage}}\\
  %%\hspace{0.02\textwidth}%
  \subfloat[Each core fast-forwards its private state and then handles
    its packet. ]{
    \begin{minipage}[t]{0.45\textwidth}
      \begin{center}
    \includegraphics[width=0.80\textwidth]{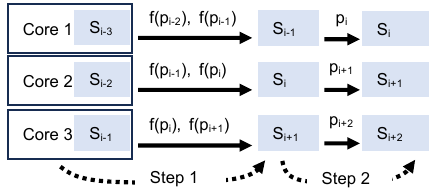}
    \label{fig:state-update}
    \vspace{-3mm}
    \end{center}
    \end{minipage}}
    \vspace{-2.5mm}
    \caption{An example illustrating the scaling principles. $p_i$ is
      the $i^{th}$ packet received by the sequencer, $f(p_j)$ are
      relevant fields from $p_j$, and $S_i$ is the state after
      processing packets $p_1, ..., p_i$ in order.}
    \label{fig:overview}
    \vspace{-6mm}
\end{figure}

\nop{
\begin{figure*}
  \centering
  \subfloat[The sequencer stores relevant fields from the packet
    history, and piggybacks the history on packets sprayed
    round-robin across cores.]{
    \begin{minipage}[t]{0.33\textwidth}
    \includegraphics[width=\textwidth]{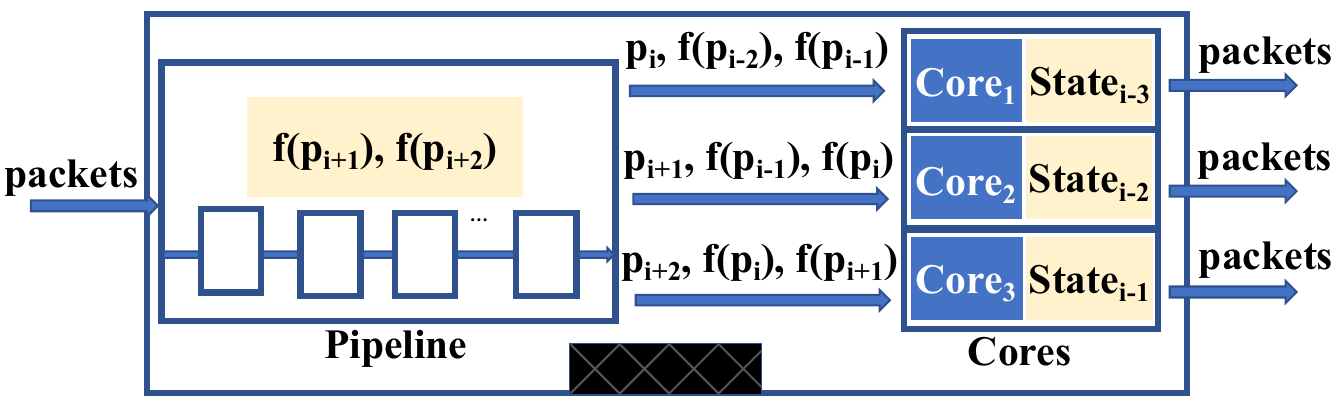}
    \label{fig:pipeline}
    \vspace{-3mm}
    \end{minipage}}
  \hspace{0.02\textwidth}%
  \subfloat[Core fast forwards its private state and then handles
    the sprayed packet.]{
    \begin{minipage}[t]{0.33\textwidth}
    \includegraphics[width=\textwidth]{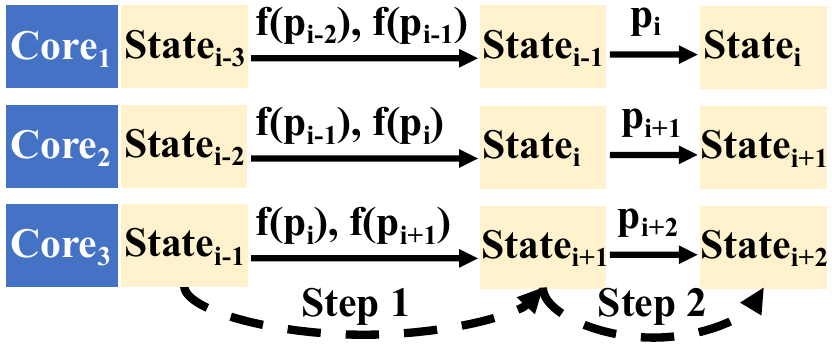}
    \label{fig:state-update}
    \vspace{-3mm}
    \end{minipage}}
    \vspace{-2.5mm}
    \caption{An example illustrating the scaling principles. $p_i$ is
      the $i^{th}$ packet received by the sequencer, $f(p_j)$ are
      relevant fields from $p_j$, and $State_i$ is the state after
      processing packets $p_1, ..., p_i$ in order.}
    \label{fig:overview}
    \vspace{-6mm}
\end{figure*}
}

%% file: pipeline-reliable-sequencer.tex
\subsection{Packet History Sequencer}
\label{sec:pipeline-sequencer}

\input{hardware-figures}

\nop{
\begin{figure}
  %%\vspace{-5mm}
  \centering
  \includegraphics[width=0.45\textwidth]{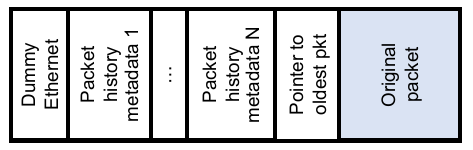}
  \vspace{-2mm}
  \caption{Packets modified to propagate history from the sequencer to
    CPU cores. The sequencer prefixes the packet history to the
    original packet, which allows for a simpler implementation in
    hardware (\Sec{pipeline-sequencer}) and simpler transformations to
    make a packet-processing program \SCR-aware
    (\App{scr-programming}). In instantiations where the sequencer is
    partly implemented on a top-of-the-rack switch
    (\Sec{operationalizing-scr}), we further prefix an additional
    Ethernet header to ensure that the NIC can process the packet
    correctly.}
  \vspace{-6mm}
  \label{fig:packet-format}
\end{figure}
}

The primary goal of the sequencer is to maintain and propagate recent
packet history to CPU cores to help replicate the computation with the
correct program state (\Sec{operationalizing-scr}). We assume the NIC
is already capable of spraying packets across CPU cores~\cite{rss,
  intel-flow-director}, and hence do not discuss that functionality
further. We describe the rest of the sequencer's functions in
terms of the following: (i) designing a packet format that modifies
existing packets to piggyback history from the sequencer to the CPU
cores; (ii) designing a hardware data structure that maintains a
recent bounded packet history at the sequencer, and enables reading
out the history into metadata on the packet. The packet fields that
are maintained in the sequencer history depend on the specific fields
used by the packet-processing application. The number of historic
packets that must be tracked depends on the degree of parallelism that
is sought, \eg the number of available CPU cores over which scaling is
implemented.

We have implemented sequencing hardware data structures on two
platforms, the Tofino programmable switch pipeline~\cite{tofino} and a
Verilog module that we integrated into the NetFPGA-PLUS
project~\cite{netfpga-plus}.

\subsubsection{Packet format}
\label{sec:packet-format}

The key question answered in this subsection is: given a packet, what
is the best place to put the packet history on it?  While this may
appear ``just an engineering detail'', designing the right packet
format has important implications to the design of hardware data
structures on the sequencer and the \SCR-aware program.

As shown in \Fig{packet-format}, we choose to place the packet history
close to the beginning of the packet, before the entirety of the
original packet. Relative to placing the packet history between
headers of the original packet, this placement simplifies the hardware
logic that writes the history into the packet, as the write must
always occur at a fixed address (0) in the packet buffer. Further, for
reasons explained in \Sec{data-structure-packet-history}, we include a
pointer to the metadata of the packet that arrived the earliest among
the ones in the piggybacked history. The earliest packet does not
always correspond to the first piece of metadata when reading the
bytes of the packet in order.

Keeping all the bytes of the original packet together in one place
also simplifies developing an \SCR-aware packet-processing
program. The packet parsing logic of the original program can remain
unmodified if the program starts parsing from the location in the
modified packet buffer which contains all the bytes of the original
packet in order.

Finally, we also prefix an additional Ethernet header to the packet in
instantiations of the sequencer which run outside of the NIC, \ie a
top-of-the-rack switch. Adding this header helps the NIC process the
packet correctly: without it, the packet appears to have an
ill-formatted Ethernet MAC header at the NIC. Our setup also uses this
Ethernet header to force RSS on the NIC~\cite{rss} to spray packets
across CPU cores (our testbed NIC (\Sec{eval-experiment-setup})
supports hashing on L2 headers). This additional Ethernet header is
not needed in a sequencer instantiation running on the NIC.

\subsubsection{Hardware data structures for packet history}
\label{sec:data-structure-packet-history}

We show how to design data structures to maintain and update a recent
packet history on two high-speed platforms, a Tofino programmable
switch pipeline~\cite{tofino} and a Verilog module integrated into the
NetFPGA-PLUS platform~\cite{netfpga-plus}. These designs are specific
to the platform where they are implemented, and hence we describe them
independently.

A key unifying principle between the two designs is that although the
items in the maintained packet history change after each packet, we
only want to update a small part of the data structure for each
packet. Conceptually, a ring buffer data structure is appropriate to
maintain such histories. Hence, in both designs, we use an {\em index
  pointer} to refer to the current data item that must be updated,
which corresponds to the head pointer of the abstract ring buffer
where data is written.

\nop{
\begin{figure}
  %%\vspace{-5mm}
  \centering
  \includegraphics[width=0.45\textwidth]{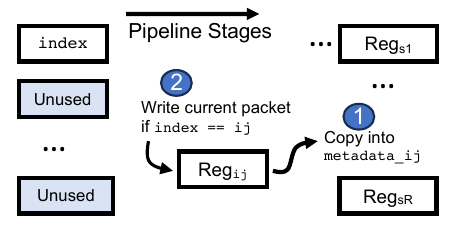}
  \vspace{-2mm}
  \caption{The data structure used to maintain and propagate packet
    history on the Tofino programmable switch pipeline
    (\Sec{data-structure-packet-history}). Inset shows the specific
    actions performed on each Tofino register.}
  \vspace{-6mm}
  \label{fig:tofino-data-structure}
\end{figure}
}

\Para{Tofino.} We use Tofino
registers~\cite{tofino-native-architecture}, which are stateful
memories to hold data on the switch, to record the bits of each
historic packet relevant to the computation in the packet-processing
program. Suppose the pipeline has $s$ match-action table stages, $R$
registers per stage, and $b$ bits per register. For simplicity in this
description, we assume there is exactly one packet field of size $b$
bits used in the computation in the packet-processing program.  Our
data structure can maintain a maximum of $(s-1) \times R \times b$
bits of recent packet history, \ie history for $(s-1) \times R$
packets, as shown in \Fig{tofino-data-structure}. We have successfully
compiled the design to the Tofino ASIC.

First, we use a single register in the first stage to store the index
pointer. The pointer refers to the specific register in the subsequent
stages that must be updated with a header field from the current
packet. The index pointer is incremented by 1 for each packet, and
wraps back to 0 when it reaches the maximum number of fields
required in the history. The pointer is also carried on a metadata
field on the packet through the remaining pipeline stages.

Next, register ALUs in subsequent stages are programmed to read out
the values stored in them into pre-designated metadata fields on the
packet. If the index pointer points to this register, an additional
action occurs: rewrite the stored contents of the register by the
pre-designated history fields from the current packet.

Finally, all the metadata fields, consisting of the packet history
fields and the index pointer, are deparsed and serialized into the
packet in the format shown in \Fig{packet-format}. We also add an
additional Ethernet header to ensure that the server NIC can receive
the packets correctly (\Sec{packet-format}).

Recent work explored the design of ring buffers to store packet
histories in the context of debugging~\cite{packet-histories-nsdi21},
reading out the histories from the control plane when a debugging
action is triggered. A key difference in our design is that reading
out histories into the packet is a data plane operation, occurring on
every packet.

\nop{
\begin{figure}
  %%\vspace{-5mm}
  \centering
  \includegraphics[width=0.45\textwidth]{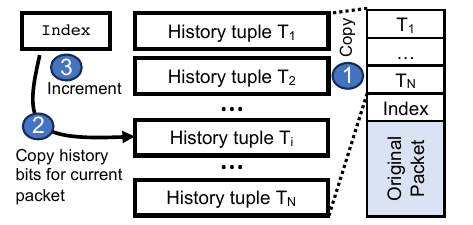}
  \vspace{-2mm}
  \caption{The data structure used to maintain and propagate packet
    history on our Verilog module integrated into NetFPGA-PLUS
    (\Sec{data-structure-packet-history}).}
  \vspace{-6mm}
  \label{fig:verilog-data-structure}
\end{figure}
}

\Para{NetFPGA.} To show the possibility of developing high-speed
fixed-function hardware for sequencing, we also present a sequencer
design developed in Verilog in \Fig{verilog-data-structure}.

Suppose we wish to maintain a history of $N$ packets, each packet
contributing $b$ bits of information. A simple design, for small
values of $N$ and $b$ (we used $N = 16$ and $b = 112$), uses a memory
which has $N$ rows, each containing a tuple of $b$ bits. We also
maintain a register containing the index pointer ($p$ bits),
initialized to zero. At the beginning the memory is initialized with
all zeroes. When a packet arrives, it is parsed to extract the bits
relevant to the packet history. Then the entire memory is read and put
in front of the packet (moving the packet contents by a fixed size
known apriori, $N \times b + p$ bits). The information relevant to the
packet history from the current packet is put into the memory row
pointed to by the index pointer, and the index pointer is incremented
(modulo the memory size). We have integrated this design into the
NetFPGA-PLUS platform.

%% file: hardware-figures.tex
\begin{figure*}
  \vspace{-5mm}
  \centering
  \subfloat[Packet format]{
    \begin{minipage}[t]{0.31\textwidth}
    \includegraphics[width=\textwidth]{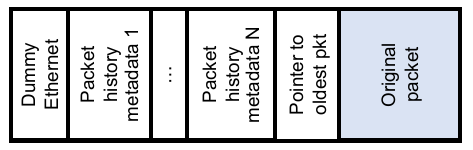}
    \label{fig:packet-format}
    \vspace{-4mm}
    \end{minipage}}
  \subfloat[Tofino sequencer]{
    \begin{minipage}[t]{0.31\textwidth}
    \includegraphics[width=\textwidth]{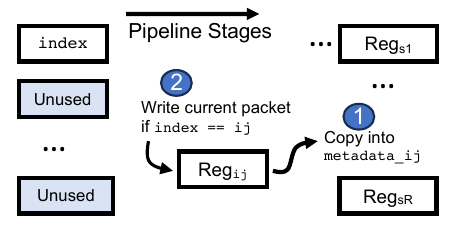}
    \label{fig:tofino-data-structure}
    \vspace{-4mm}
  \end{minipage}}
  \subfloat[RTL sequencer]{
    \begin{minipage}[t]{0.31\textwidth}
    \includegraphics[width=\textwidth]{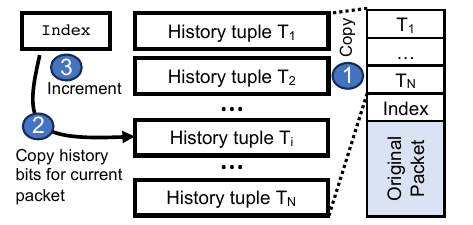}
    \label{fig:verilog-data-structure}
    \vspace{-4mm}
  \end{minipage}}
  %%\vspace{-2.5mm}
  \caption{Hardware data structures.
    (a) Packets modified to propagate history from the sequencer to
    CPU cores. The sequencer prefixes the packet history to the
    original packet, which allows for a simpler implementation in
    hardware (\Sec{pipeline-sequencer}) and simpler transformations to
    make a packet-processing program \SCR-aware
    (\App{scr-programming}). In instantiations where the sequencer is
    partly implemented on a top-of-the-rack switch
    (\Sec{operationalizing-scr}), we further prefix a dummy Ethernet
    header to ensure that the NIC can process the packet correctly.
    (b) The data structure used to maintain and propagate packet
    history on the Tofino programmable switch pipeline
    (\Sec{data-structure-packet-history}). Inset shows the specific
    actions performed on each Tofino register.
    (c) The data structure used to maintain and propagate packet
    history on our Verilog module integrated into NetFPGA-PLUS
    (\Sec{data-structure-packet-history}).
    }
    \label{fig:hardware-data-structures}
    \vspace{-6mm}
\end{figure*}

%% file: packet-loss-nondeterminism.tex
\subsection{Packet Loss and Nondeterminism}
\label{sec:packet-loss-nondeterminism}

\Para{Handling Packet Loss.}  {\revaddlargebegin Packets can be lost
  either prior to the sequencer, after the sequencer but prior to
  processing at a CPU core, or after processing at a core. Among
  these, we only care about the second kind of packet loss, since this
  is the only one that is problematic specifically for \SCR,
  introducing the possibility that flow states on the CPU cores might
  become inconsistent with each other.

  We expect that \SCR will be deployed in scenarios where the event of
  packet loss between the sequencer and CPU cores is rare. First, we
  do not anticipate any packet loss in an instantiation where the
  sequencer is running entirely on a NIC, since the host interconnect
  between the NIC and CPU cores uses credit-based flow control and is
  lossless by design~\cite{host-congestion-hotnets22,
    pciebench-sigcomm18}.  In an instantiation where the sequencer is
  running on a top-of-the-rack switch, it is possible to run
  link-level flow control mechanisms like PFC~\cite{pfc-ieee-8021qbb}
  (as some large production networks do~\cite{rdma-at-scale}) to
  prevent packet loss between the switch and server
  cores. We discuss below how \SCR can handle rare
  packet drop and corruption events while maintaining consistency
  among the states on CPU cores.

  How should a CPU core that has lost a packet arriving from the
  sequencer synchronize itself to the correct flow state?  There are
  two design options: the core can either explicitly read the full
  flow state from a more up-to-date core, or it can read the packet
  history from either the sequencer or a log written by a more
  up-to-date core, and then use the history to catch up its private
  state (akin to \Fig{overview}). Since we operate in a regime where
  packet losses are rare, but the full set of flow states is large, we
  prefer to synchronize the packet history rather than the
  state. Further, to simplify the overall design, we avoid explicit
  coordination between the cores and the sequencer, synchronizing
  the history among the cores only.

  Our objective is {\em atomicity:} any packet is either processed by
  all the cores or none of the cores. If a packet is sent by the
  sequencer to {\em any} core (in original or as part of the packet
  history), it should be processed in the correct order by {\em all}
  the cores.  To achieve this, we (i) have the sequencer attach an
  incrementing sequence number to each packet released by it; (ii) use
  a per-core, lockless, single-writer multiple-reader log, into which
  each core writes the history contained in each packet it receives
  (including the relevant data for the original packet); and (iii)
  introduce an algorithm to catch up the flow state on each core upon
  detection of loss.

  The algorithm proceeds as follows (more information is available in
  \Alg{loss_recovery} in \App{loss-recovery-algorithm}). Each CPU core
  $c$ maintains a per-core log with one entry for each sequence number
  $i$. In a system with $N$ cores, the history metadata of the packet
  with sequence $i$ (say $history[i]$) will appear in packets with
  sequence numbers $i$ through $i + N - 1$; conversely, a packet with
  sequence number $j$ contains $history[minseq], \cdots, history[j]$
  where $minseq \triangleq max(1, j-N+1)$.

  For core $c$ and sequence number $i$, $log[c][i] :=$
  \begin{gather*}
  \begin{cases}
    history[i] & \text{if history for sequence $i$ was received at $c$} \\
    NOT\_INIT & \text{if the highest sequence received at $c$ is $j < i$} \\
    LOST & \text{if $c$ has received sequence $j > i$, }\\
    & \text{but sequence $i$ was not received at $c$} \\
  \end{cases}
  \end{gather*}

Initially, $\forall c, i: log[c][i] \triangleq NOT\_INIT$, to denote
that the log entry for every sequence $i$ is uninitialized at every core
$c$. When a (fixed) core $c$ receives a packet with sequence number
$j$, it first detects packet loss by comparing the max sequence number
it has seen so far (say $max[c]$) with the earliest sequence number in
the new packet it receives ($minseq$), assuming no reordering between
the sequencer and the core. Then, $c$ processes every sequence number
$k \in \{k\ |\ max[c] < k \leq j\}$ in order of increasing $k$, as
follows:

\begin{CompactEnumerate}

\item if $k < minseq$, \ie sequence $k$ was lost between the sequencer
  and core $c$, the core updates $log[c][k] \leftarrow LOST$.  For
  such packets, core $c$ will read from the logs of other cores $c'
  \neq c$ in a loop, until $c$ discovers either that (i) $history[k]$
  is written in $log[c'][k]$, in which case $c$ catches up its private
  state by reading this history; or (ii) $log[c'][k] = LOST$ on all
  cores $c' \neq c$, concluding that sequence $k$ was never originally
  received on any core, and does not need to be recovered for
  atomicity;

\item if $minseq \leq k \leq j$, \ie sequence $k$ is successfully
  received at core $c$ (as part of the current packet), core $c$
  updates $log[c][k] \leftarrow history[k]$ available in the packet,
  and then proceeds with regular processing as in
  \Sec{operationalizing-scr}.

\end{CompactEnumerate}

In \App{loss-recovery-algorithm}, we formally prove that, under some
mild assumptions, this algorithm always terminates in a state that is
eventually consistent across all CPU cores. Despite cores possibly
waiting on one another, there will be no deadlocks.
\forcameraready{that every core will eventually receive a new packet
  from the sequencer.}

%% For $minseq \leq k \leq j$, core $c$ updates $log[c][k] \leftarrow
%% history[k]$ available in the packet, and if $max_c < j-N$, the core
%% also updates $log[c][k] \leftarrow LOST$ for $\max_c < k < j-N+1$.

%% Starting from the smallest sequence number $k$ that is $LOST$ at $c$,
%% core $c$ will read from the logs of other cores $c' \neq c$ in a loop,
%% until $c$ discovers either that (i) $history[k]$ is written in
%% $log[c'][k]$, in which case $c$ catches up its private state by
%% reading this history; or (ii) $log[c'][k] = LOST$ on all cores $c'
%% \neq c$, concluding that the history for sequence $k$ was never
%% originally received on any core, and does not need to be processed or
%% recovered for consistency. In \App{loss-recovery-algorithm}, we
%% formally prove that this algorithm always terminates in a
%% state that is eventually consistent across all CPU cores under the
%% mild assumption that every core will eventually receive a new packet.
}

\Para{Handling non-deterministic programs.} {\revaddlargebegin
  Programs replicating computation across cores may diverge from each
  other due to two possible concerns. The first is the use of
  timestamps in computations (\eg to implement a token bucket rate
  limiter). We avoid the use of timestamps measured
  locally at CPU cores, in favor of measuring and attaching timestamps
  to the packet history at the sequencer, and having the
  packet-processing program use the attached timestamp instead.  Modern NICs and
  programmable switches support high-resolution hardware
  timestamping over packets~\cite{tofino-native-architecture,
    nvidia-connectx-6}. The second concern is the use of random
  numbers in computations. We make program execution deterministic by
  fixing the seed of the pseudorandom number generator used across
  cores. 
}

%% file: evaluation.tex
\section{Evaluation}
\label{sec:evaluation}

\input{eval-table-application-characteristics}

We seek to answer two main questions through the experiment setup
described in \Sec{eval-experiment-setup}.

\noindent (1) Does state-compute replication provide better multi-core
scaling than existing techniques (\Sec{eval-throughput-scaling})?

\noindent (2) How practical is sequencer hardware
(\Sec{eval-hardware-resource-use})?

\subsection{Experiment Setup}
\label{sec:eval-experiment-setup}

\Para{Machines and configurations.} Our experiment setup consists of
two server machines connected back-to-back over a 100 Gbit/s
Nvidia/Mellanox ConnectX-5 NIC on each machine. Our servers run Intel
Ice Lake processors (Xeon Gold 6334) with 16 physical cores (32
hyperthreads) and 256 GB DDR4 physical memory spread over two NUMA
nodes. The system bus is PCIe 4.0 16x. We run Ubuntu 22 with Linux
kernel v6.5. One machine serves as a packet replayer/generator,
running a DPDK burst-replay program which can transmit packets from a
provided traffic trace. We have tested that the traffic generator can
replay large traces (1 million packets) at speeds of $\sim$ 120
million packets/second (Mpps), for sufficiently small packets (so that
the NIC bandwidth isn't saturated first). The traffic generator can be
directed to transmit packets at a fixed transmission (TX) rate and
measure the corresponding received (RX) packet rate. Our second server
is the Device Under Test (DUT), which runs on identical hardware and
operating system as the first server. We implement standard
configurations to benchmark high-speed packet
processing~\cite{xdp-conext18}: hyperthreading is disabled; the
processor C-states, DVFS, and TurboBoost are disabled; dynamic IRQ
balancing is disabled; and the clock frequency is set to a fixed 3.6
GHz. We enable PCIe descriptor compression and use 256 \nop{in-flight}
PCIe descriptors. Receive-side scaling (RSS~\cite{rss}) is configured
according to the baselines/programs, see
\Tab{application-characteristics}. We use a single receive queue (RXQ)
per core unless specified otherwise.

\Para{The definition of throughput.} We use the standard {\em maximum
  loss-free forwarding rate} (MLFFR~\cite{rfc2544}) methodology to
benchmark packet-processing throughput. Our threshold for packet loss
is in fact larger than zero (we count $< 4\%$ loss as ``loss-free''),
since, at high speeds we have observed that the software typically
always incurs a small amount of bursty packet loss. We use binary
search to expedite the search for the MLFFR, stopping the search when
the bounds of the search interval are separated by less than 0.4
Mpps. \nop{We ensure that the computations are bottlenecked by CPU in
  all of our throughput measurement experiments.} Experimentally, we
observe that MLFFR is a stable throughput metric: we get highly
repeatable results across multiple runs. We only report throughput
from a single run of the MLFFR binary search.

\Para{Traces.} We are interested in understanding whether \SCR
provides better multi-core scaling than existing techniques on
realistic traffic workloads. We have set up and used three traces for
throughput comparison: a university data center
trace~\cite{microsoft-network-sigcomm10}, a wide-area Internet
backbone trace from CAIDA~\cite{caida}, and a synthetic trace with
flows whose sizes and inter-arrivals were sampled from a hyperscalar's
data center flow characteristics~\cite{dctcp-sigcomm10}. These traces
are highly dynamic, with flow states being created and destroyed
throughout---an aspect that we believe is crucial to handle in real
deployment environments (\ie the programs are not simply processing a
stable set of active flows). Further, we ensure that all TCP flows
that begin in the trace also end, by setting TCP SYN and FIN flags for
the first and last packets (resp.) of each flow in the trace. This
allows the trace to be replayed multiple times with the correct
program semantics. The flow size distributions of the traces are shown
in \Fig{flow-size-distributions}.

\input{eval-cdf-plots.tex}

The eBPF framework limits our implementations in terms of the number
of concurrent flows that our stateful data structures can
include. This is not a limitation of the techniques, but an
artifact of the current packet-processing framework we
use (eBPF/XDP). To account for this limitation, specifically for the
CAIDA trace, we have sampled flows from the trace's empirical flow
size distribution to faithfully reflect the underlying distribution,
without over-running the limit on the number of concurrent flows that
any of our baseline programs may hold across the lifetime of the
experiment.

\Para{Baselines.} We compare state-compute replication against (i)
state sharing, an approach that uses hardware atomic instructions
(when the stateful update is simple enough) or eBPF
spinlocks~\cite{ebpf-spinlocks} (when it is not) to share state across
CPU cores; (ii) state sharding using classic RSS; and (iii) sharding
using RSS++~\cite{rss++-conext19, automatic-parallelization-nsdi24},
the state-of-the-art flow sharding technique to balance CPU
load. RSS++ solves an optimization problem that takes as input the
incoming load imposed by flow shards, and migrates shards to minimize
a linear combination of load imbalance across CPU cores and the number
of cross-core shard transfers needed.
%
%%  It is
%% expected to scale better than RSS by rebalancing flows across cores by
%% measuring and responding to high per-core workload.
Both \SCR and state sharing spray packets evenly across CPU cores. The
packets sent to each core for the sharding techniques depends on the
configuration of RSS, which varies across the programs we
evaluated (see below and \Tab{application-characteristics}). Running
RSS++ over eBPF/XDP requires patching the NIC
driver~\cite{rss++-kernel-patch}. \forcameraready{to expose the RX
  hash on packets to XDP programs.} Unless specified otherwise, we run \SCR
without loss recovery
(\Sec{packet-loss-nondeterminism}); \forcameraready{as we believe this
  is the most representative scenario in which \SCR will be deployed.}
we evaluate loss recovery separately (\Sec{eval-throughput-scaling}).

\Para{Programs.} We tested five packet-processing programs
developed in eBPF/XDP, including (i) a heavy hitter monitor, (ii) DDoS
mitigator, (iii) TCP connection state tracker, (iv) port-knocking
firewall, and (v) a token bucket policer.
\revadd{\Tab{application-characteristics} summarizes these programs. Each
program maintains state across packets in the form of a key-value
dictionary, whose size and contents are listed in the table.} We
developed a cuckoo hash table to implement the functionality of this
dictionary with a single BPF helper call~\cite{bpf-helpers}.
%
%% and use it across all the baselines (sharding, sharing, \SCR \qx{we
%% didn't use it for sharing, as it might be inefficient: sharing
%% requires lock the state, and we have to lock the enire cuckoo map,
%% as BPF only allows to use spinlock for one BPF map element (ie, one
%% cucokoo map) but not inside an element.}).
%% TODO: add a footnote to clarify this after finding space.
%
The packet fields in the key determine how RSS must be configured:
packets having the same key fields must be sent to the same CPU
core. However, today's NICs do not allow RSS to steer packets on
arbitrary sets of packet
fields~\cite{automatic-parallelization-nsdi24}. We pre-process the
trace to ensure that RSS hashing indeed shards the flow state
correctly. \forcameraready{For example, on the NIC in
  our testbed, the source ({\ct srcip}) and destination ({\ct dstip})
  IP addresses may be used together, but not separately, as the RSS
  key to hash a packet to a core. For a program that maintains flow
  state at the granularity of {\ct dstip}, an RSS key ({\ct srcip,
    dstip}) could steer two packets with the same {\ct dstip} but
  different {\ct srcip} to different CPU cores, violating sharding at
  the granularity of {\ct dstip}. To prevent this, and to evaluate our
  sharding baselines fairly, we pre-process our traces (\eg modifying
  packets such that every {\ct srcip, dstip} combination in the trace
  hashes to a core that only depends on {\ct dstip}) to ensure that
  RSS hashing indeed shards the flow state correctly.}
%
%% To be fair to
%% sharding-oriented baselines, we process our input traces to ensure
%% that a superset of the key fields may be used to configure RSS to
%% steer packets correctly for our trace, indeed sharding the program
%% state. 
For the connection tracker, since both directions of the
connection must go to the same CPU core, we use the keyed hash
function prescribed by symmetric RSS~\cite{symrss-TR}.

\subsection{Multi-Core Throughput Scaling}
\label{sec:eval-throughput-scaling}

In this section, we compare the MLFFR throughput
(\Sec{eval-experiment-setup}) of several packet-processing programs
(\Tab{application-characteristics}) scaled across multiple cores using
four techniques: SCR (\Sec{design}), state sharing with
packets sprayed evenly across all cores, sharding using RSS, and
sharing using RSS++
(\Sec{motivation-background}). Since the TCP connection tracking
program requires packets from the two directions of the connection
to be aligned, we evaluated it on a synthetic but
realistic hyperscalar data center trace
(\Sec{eval-experiment-setup}). For the rest of the programs, we
report results from real university data center and Internet backbone
traces.

\input{eval-trace-throughput-graphs}

We have ensured that these experiments reflect a fair comparison of
CPU packet-processing efficacy. First, we truncated the packets in the
traces to a size smaller than the full MTU, to stress CPU performance
with a high packets/second (Mpps) workload
(\Sec{scr-principle}). Further, we fix the packet sizes used across
all baselines for a given program. \nop{since feeding packets of
  different sizes to the program for the same fixed packets/second
  arrival rate may induce bottlenecks other than the CPU (we show such
  experiments later).}  We used a fixed packet size of 256 bytes for
the connection tracker and 192 bytes for the others. The packet size
limits the number of items of history metadata that can be piggybacked
on each packet. Since the metadata size changes by the program
(\Tab{application-characteristics}), the maximum number of cores we
can support for a fixed packet size also varies by program (we
support 7 cores for the token bucket, heavy hitter detector, and
connection tracker, and 14 for the DDoS mitigation and port-knocking
firewall).

\Para{Throughput results.} 
\Fig{throughput-traces} and \Fig{throughput-conntrack-dctcp} show the
throughput as we increase the number
of packet-processing cores. \SCR is the only multi-core scaling
technique that can monotonically scale the throughput of all the stateful
packet-processing programs we evaluated across multiple cores,
regardless of the flow size distribution (\Sec{goals}). The throughput
for \SCR increases
linearly across cores in all of the configurations we tested. Somewhat
surprisingly, \SCR provides even better absolute performance than
hardware atomic instructions in the case of the heavy hitter
and DDoS mitigation programs. However, the performance of lock-based
sharing falls off catastrophically with 3 or more cores.

\input{eval-trace-conntrack}

The throughput of sharding using RSS depends on the vagaries of how
the RSS hash function steers flows to cores. RSS can neither split a
single elephant flow, nor does it intelligently redistribute elephant
flows to balance load across CPU cores. RSS++ indeed redistributes
flows across cores to adapt to high workload on some cores.  Such
redistribution can sometimes confer benefits over RSS (for example,
see \Fig{throughput-traces}(e) and (h)).  However, fundamentally,
flow-based sharding is ineffective when the workload is highly skewed
with elephant flows, as many real traffic workloads are
(\Fig{flow-size-distributions}).  Moreover, our results show that
RSS++ is not always better than simple RSS.  Re-balancing load by
migrating a flow shard across cores requires bouncing the cache
line(s) containing the flow states across cores, an action that can
impede high performance if done too frequently. On the other hand,
rebalancing infrequently may reintroduce skew across cores, since the
future load from flow shards may drift from their current load
(the basis for load rebalancing in RSS++).

%% NG: We no longer have a separate RSS++ comparison (Yay!)
%
% \input{eval-rss++}
%% \Para{Comparison with RSS++.} We compared \SCR against
%% RSS++~\cite{rss++-conext19} to evaluate whether a state-of-the-art
%% sharding solution scales better by rebalancing work across cores upon
%% heavy load.  Due to a configuration issue in RSS++, we could only run
%% RSS++ on CPU cores that were on a NUMA node different from that of the
%% NIC. We were only able to compare \SCR and RSS++ both running on these
%% far-NUMA cores, with one program, TCP connection
%% tracking. \Fig{throughput-rsspp} shows the results.  Fundamentally,
%% flow-level sharding is ineffective when the workload is highly skewed
%% with elephant flows, as many real traffic workloads tend to be
%% (\Fig{flow-size-distributions}).

\Para{Why does \SCR scale better than the other techniques?}
\Fig{pcm-explanations} shows detailed performance metrics measured
from Intel's performance counter monitor (PCM~\cite{intel-pcm}) and
BPF profiling~\cite{bpf-profiling}. We measure the L2 cache hit
ratios, instructions retired per cycle (IPC), and the program's
computation latency (only the XDP portion, excluding the \dispatch
functionality in the driver), as the load offered to the system
increases, when our token bucket policer is run across different
numbers of cores (2, 4, or 7). The numbers show the averages for these
metrics across the cores running the program. Error bars for IPC show
the min and max values across cores.
\forcameraready{IPC is a meaningful metric to
evaluate ``CPU goodput'' for eBPF/XDP programs: unlike high-speed
packet processing frameworks like DPDK which poll the NIC and exhibit
persistently high IPC~\cite{routebricks-sosp09}, eBPF/XDP drivers
adapt CPU usage to load through a mix of polling and interrupts.}

\input{eval-pcm-latency}

Lock-based sharing in general suffers from lower L2 cache hit ratios
((a)--(c)) and higher latencies ((g)--(i)) due to lock and cache line
contention across cores---a trend that holds as the offered load
increases and also with additional cores at the same offered load. As
we might expect, IPC increases with the offered load ((d)--(f)), since
the cores get busier with packet processing. While the sharding
approaches (RSS and RSS++) effectively use the CPU with a high average
IPC for 2 cores, their average IPC values drop significantly with
additional cores, with large variation across cores (see error bars).
This indicates a severe imbalance of useful CPU work across cores:
Flow-affinity-based sharding approaches are unable to balance packet
processing load effectively across cores, leaving some cores idle and
others heavily used.  In contrast, \SCR has a consistently high IPC
with more cores and higher offered loads. \SCR has higher
packet-processing latency ((g)--(i)) than RSS-sharding since it needs
to process the history for each packet. \forcameraready{RSS++
  sometimes incurs higher compute latency than \SCR due to its need to
  monitor per-shard load, which requires additional memory
  operations.} However, its more effective usage of the CPU cores
results in better throughput (\Fig{throughput-traces}(g)).

\input{eval-limits-scr}

\Para{Limits to \SCR scaling.}
{\revaddlargebegin
\SCR suffers from two kinds of scaling limitations. 
First, as discussed in \Sec{scr-principle}, as the compute latency
increases in comparison to the dispatch latency, the effectiveness of
\SCR's multi-core scaling reduces.
\forcameraready{This is because more time is spent
``catching up'' state, incurring significant duplicated work across
CPU cores.}
We evaluate how the throughput of a stateless program varies as the
compute latency of this program increases, shown both in
packets/second (\Fig{throughput-compute-time-increase-absolute-1rxq},
\Fig{throughput-compute-time-increase-absolute-2rxq}) and
normalized to the single-core throughput for that compute latency
(\Fig{throughput-compute-time-increase-normalized}). With a small
compute latency (left of the graph), using $N$ cores provides
$\approx N \times$ throughput relative to a single core, but this
relative benefit diminishes with increasing compute latency.
} %% end revaddlargebegin

{
\revaddlargebegin
Second, \SCR's attachment of histories to packets incurs
non-negligible byte/second overheads.
%% even if the packet/second throughput scales better.
%
Adding to the number of bytes per packet increases L3 cache pressure
due to higher DDIO cache occupancy~\cite{intel-ddio} and incurs
additional PCIe transactions and
bandwidth~\cite{pciebench-sigcomm18}. Further, when packet histories
are appended outside the NIC (\eg top-of-the-rack switch), \SCR may
saturate the NIC earlier than other approaches.  %% None of the
%% multi-core scaling techniques other than \SCR incur these consequences
%% of the byte/second overhead.
%
%% To evaluate how \SCR scales when accounting for this overhead, \ie
%% saturating the NIC earlier than other techniques and plateauing
%% earlier in packets/second processed,
%
We compare \SCR against the shared and sharded approaches, when \SCR
alone adds history metadata before packets are fed into the NIC while
the packets for the latter approaches are truncated to 64 bytes.
\Fig{throughput-numbers-assorted}(a) shows the throughput
of the token bucket program with the university data center
trace.
\nop{
Increasing
the packet size specifically for \SCR has at least three performance
consequences: it may saturate the NIC earlier for \SCR than other
baselines, it may result in additional L3 cache memory pressure due to
higher DDIO cache occupancy~\cite{intel-ddio}, and it may result in
PCIe bottlenecks~\cite{pciebench-sigcomm18} due to additional PCIe
transactions and consumption of PCIe bandwidth.
} % ends nop
After 11 cores, the CPU is no longer the bottleneck for \SCR. \nop{as the
memory bandwidth to cores on the far NUMA is saturated.} This prevents
\SCR from scaling to a higher packets/second throughput. Yet, \SCR
saturates at a throughput much higher than the other techniques.
}

\input{eval-loss-recovery}

\Para{Overheads of \SCR's loss recovery handling.}  \revadd{We
  evaluate how \SCR's loss recovery algorithm impacts throughput with
  and without packet loss. \Fig{throughput-numbers-assorted}(b)
  compares a version of \SCR without incorporating the loss recovery
  algorithm, against a version that incorporates loss recovery at
  different artificially-injected random packet losses (0\%, 0.01\%,
  0.1\% and 1\%).  We also show the performance of existing scaling
  techniques (shared state, RSS, RSS++). The mere inclusion of the
  loss recovery algorithm impacts performance due to the additional
  logging operations.  Moreover, \SCR's throughput degrades with
  higher loss rate due to recovery-related synchronization
  (\Sec{packet-loss-nondeterminism}). However, \SCR still outperforms
  and outscales existing multi-core scaling techniques.  }

%% %% Possible Experiments to Check:
%% throughput with realistic flow size distributions
%% PCM explanations for ``why''
%% comparisons against rss++
%% Other smaller experiments:
%% - smaller number of flows/requests
%% - single flow scaling
%% - dependence on the flow size distribution?
%% any impact of packet reordering?
%% scaling limits
%% showing mainly: MLFFR experiments with fixed ``internal'' packet rate
%% later, consider possible MLFFR experiments with fixed external packet
%% rate ... especially with 200 Gbit/s nic
%% latency numbers?
%% numbers to show balanced CPU utilization
%% 2.2 motivational graphs ideas (gianni): how contention affects
%% performance for shared state, how cpu imbalance created for
%% rss/sharding in caida.

\subsection{Practicality of Sequencer Hardware}
\label{sec:eval-hardware-resource-use}

\input{eval-hardware-synthesis-numbers}

%%\vspace{-0.5cm}

We integrated our Verilog module implementing the sequencer
(\Sec{data-structure-packet-history}) into the
NetFPGA-PLUS~\cite{netfpga-plus} reference switch, which is clocked at
250 MHz with a data bus width of 1024 bits. We use the Alveo U250
board, which contains 1728000 lookup tables (LUTs) and 3456000
flip-flops.

We synthesized our sequencer design with different numbers of memory
rows (\Sec{data-structure-packet-history}), corresponding to the size
of the packet history (in number of packets). 
Each row is 112 bits long, enough to maintain a TCP 4-tuple and an
additional 16-bit value (\eg a counter, timestamp, \etc) for each
historic packet.  %% This is sufficient for many of the programs we
%% have tested (\Tab{application-characteristics}).
%
\Tab{verilog-synthesis-results} shows the resource usage. Our design
meets timing at 250 MHz, implying an achievable bandwidth of more than
200 Gbit/s. If each packet history metadata in the program is smaller
than a row (112 bits),\forcameraready{ parallelizing across $N$ cores
  requires $N$ rows; for such programs,} our design can meet timing
while scaling to 128 cores. The LUT and flip-flop hardware usage is
negligible compared to the FPGA capacity. \forcameraready{at all row
  counts measured.} We believe that our sequencer design may be
simple and cheap enough to be added as an on-chip accelerator
to a future NIC.

\revadd{ We have also implemented a stateful-register-based design of
  the sequencer on the Tofino programmable switch
  (\Sec{data-structure-packet-history}). The results are shown in
  \Tab{tofino-synthesis-results}. Our implementation was designed to
  use as many stateful registers and ALUs as possible (our design uses
  93\% on average across stages) to hold the largest number of bits of
  packet history.  Our design holds 44 32-bit fields, sufficient to
  parallelize (\Tab{application-characteristics}) the DDoS mitigator
  over 44 cores, the port-knocking firewall over 22 cores, the heavy 
  hitter and token bucket over 9 cores, or the connection tracker over
  5 cores. The small number of stateful ALUs on the platform, as well
  as the limit on the number of bits that can be read out from stateful
  memory into packet fields, restrict the Tofino sequencer from scaling
  to a larger number of CPU cores.  }

\input{eval-tofino-resource-numbers}

%% file: eval-table-application-characteristics.tex
\begin{center}
  \vspace{-4mm}
  \begin{table*}
  \resizebox{\textwidth}{!}{    
  \begin{tabular}{|l|c|c|c|c|c|c|c|}
    \hline
    {\bf Program} & \multicolumn{2}{|c|}{{\bf State}} & {\bf Metadata size} & {\bf RSS hash} & {\bf Packet traces} & {\bf Atomic HW} & {\bf Lines of code}\\
    & {\bf Key} & {\bf Value} & {\bf (bytes/packet)} & {\bf fields} & {\bf evaluated} & {\bf vs. Locks} & {\bf (shard/RSS)} \\
    \hline
    %% DDoS mitigator & SRC IP & count & 4 & SRC IP & CAIDA, UnivDC & Atomic HW & 168 \\
    %% Heavy hitter monitor & 5-tuple & flow size & 17 & 5-tuple & CAIDA, UnivDC & Atomic HW & 141 \\
    %% TCP connection state tracking & 5-tuple & conn state, ts, seq, ip/port rev & 30 & 5-tuple & hyperscalar data center trace & Locks &  1029 \\
    %% Token bucket policer & 5-tuple & last packet ts, \# tokens & 17 & 5-tuple & CAIDA, UnivDC & Locks & 169 \\
    %% Port-knocking firewall & SRC IP & knocking status (eg. Open) & 4 & SRC IP & CAIDA, UnivDC & Locks & 123 \\   
    DDoS mitigator & source IP & count & 4 & src \& dst IP & CAIDA, Univ DC & Atomic HW & 168 \\
    Heavy hitter monitor & 5-tuple & flow size & 18 & 5-tuple & CAIDA, Univ DC & Atomic HW & 141 \\
    TCP connection state tracking & 5-tuple & TCP state, timestamp, seq \# & 30 & 5-tuple & Hyperscalar DC & Locks & 1029 \\
    Token bucket policer & 5-tuple & last packet timestamp, \# tokens & 18 & 5-tuple & CAIDA, UnivDC & Locks & 169 \\
    Port-knocking firewall & source IP & knocking state (\eg\ {\ct OPEN}) & 8 & src \& dst IP & CAIDA, UnivDC & Locks & 123 \\   
    \hline
  \end{tabular}
  }
  \caption{The packet-processing programs we evaluated.}
  \label{tab:application-characteristics}
  \end{table*}
  \vspace{-4mm}
\end{center}

%% file: eval-cdf-plots.tex
\begin{figure}
  \vspace{-5mm}
  \centering
  \subfloat[University DC]{
    \begin{minipage}[t]{0.16\textwidth}
    \includegraphics[width=\textwidth]{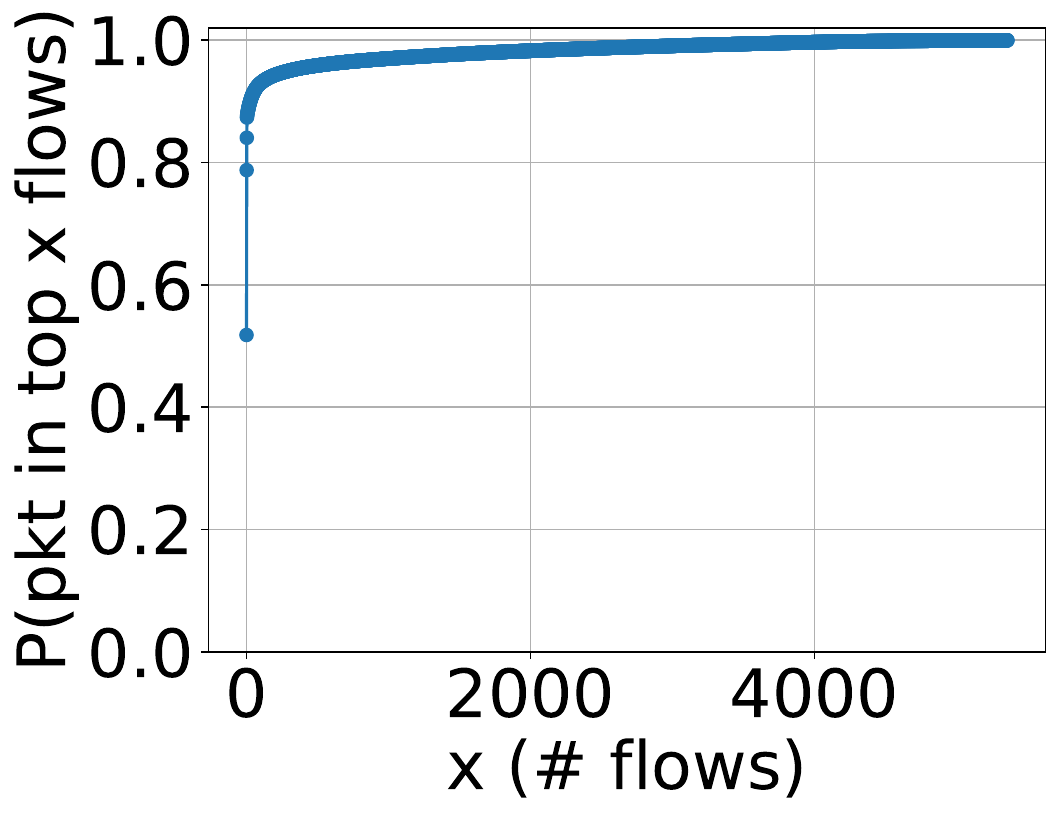}
    \label{fig:cdf-dc-trace}
    %%\vspace{-4mm}
    \end{minipage}} 
    % \hspace{0.01\textwidth}%
  \subfloat[Internet backbone]{
    \begin{minipage}[t]{0.16\textwidth}
    \includegraphics[width=\textwidth]{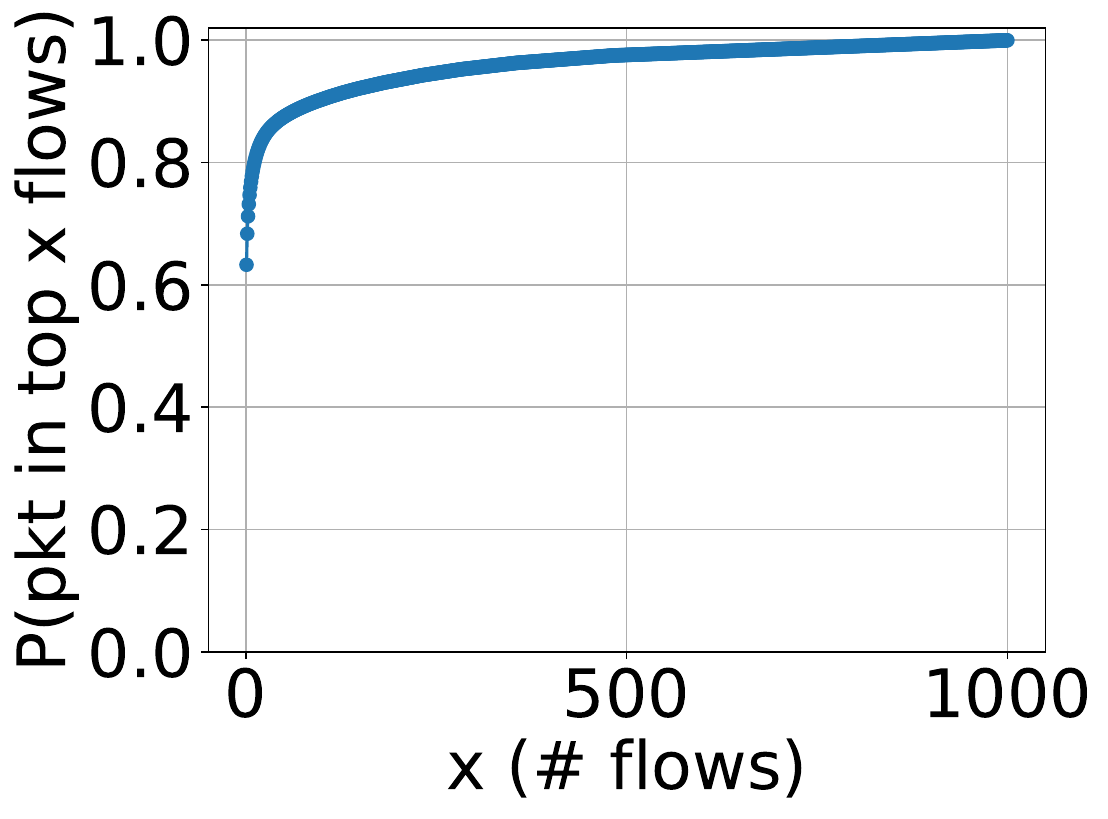}
    \label{fig:cdf-caida-trace}
    %%\vspace{-4mm}
  \end{minipage}}
  \subfloat[Hyperscalar DC]{
    \begin{minipage}[t]{0.16\textwidth}
    \includegraphics[width=\textwidth]{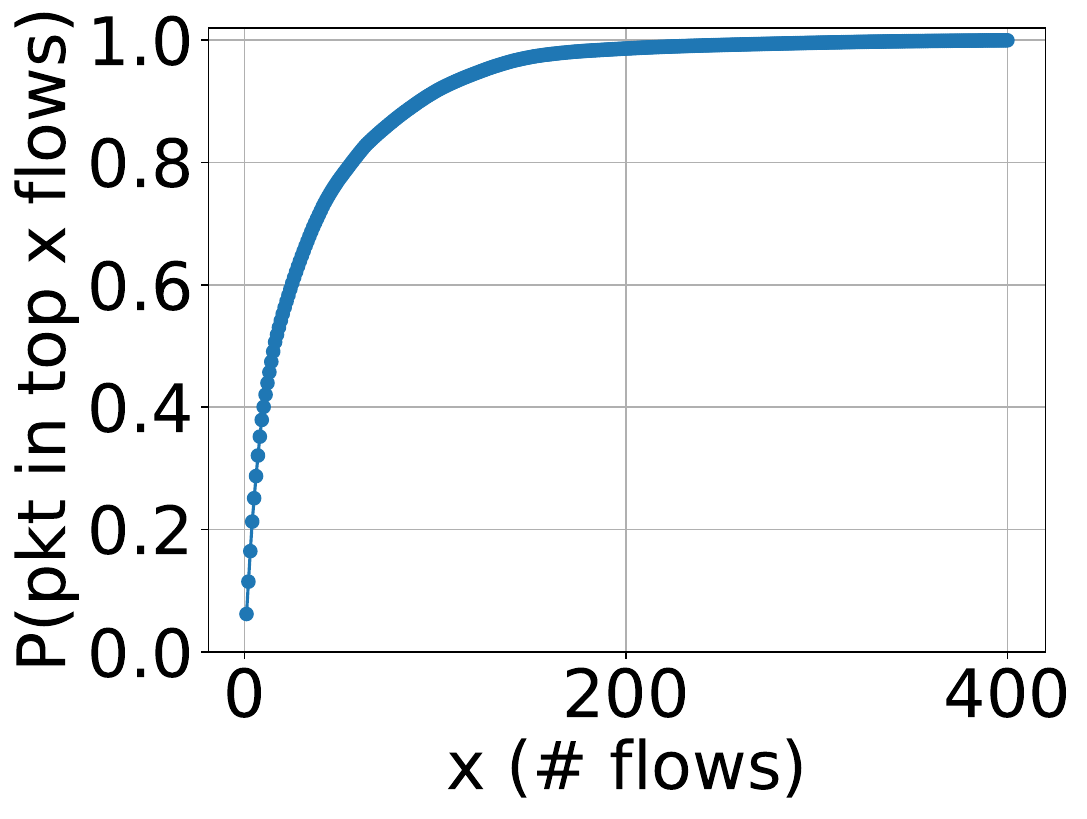}
    \label{fig:cdf-dctcp-trace}
    %%\vspace{-4mm}
  \end{minipage}}
    %%\vspace{-2.5mm}
    \caption{Flow size distributions of the packet traces we used. We
      used real packet traces captured at (a) university data
      center~\cite{microsoft-network-sigcomm10} and (b) wide-area
      Internet backbone by CAIDA~\cite{caida}. We also synthesized (c)
      a packet trace with real TCP flows whose sizes are drawn from
      Microsoft's data center flow size
      distribution~\cite{dctcp-sigcomm10}. }
    \label{fig:flow-size-distributions}
    \vspace{-5mm}
\end{figure}

%% file: eval-trace-throughput-graphs.tex
\begin{figure*}
  %\vspace{-5mm}
  \centering
  \includegraphics[width=\textwidth]{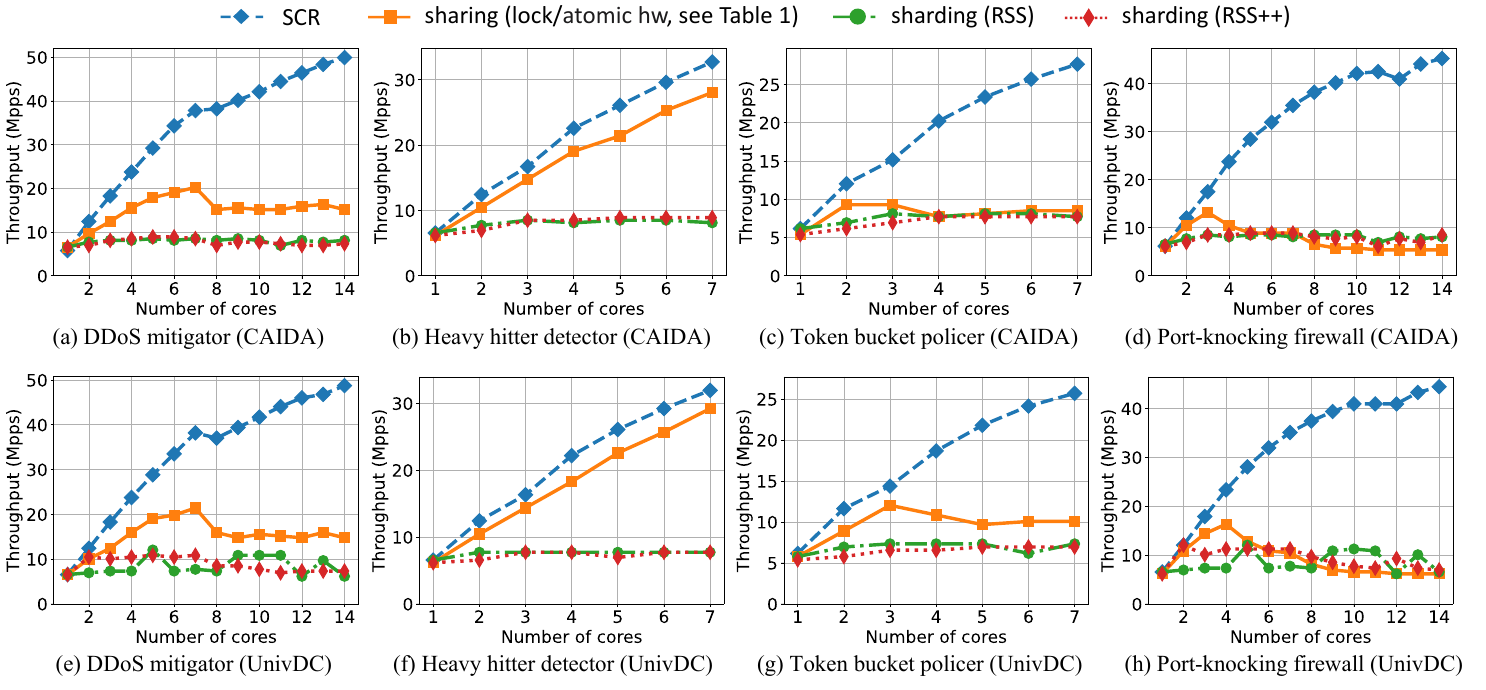}
  %%\vspace{-2.5mm}
  \caption{Throughput (\Sec{eval-experiment-setup}) in millions of
    packets per second (Mpps) of four stateful packet-processing
    programs implemented using state-compute replication
    (\Sec{design}), shared state, and sharding
    (\Sec{motivation-background}).
    %% Points denote averages and error
    %% bars denote standard deviations across 3 runs.
    Packet traffic is replayed from real data center and Internet backbone
    traces.}
  \label{fig:throughput-traces}
  \vspace{-5mm}
\end{figure*}

%% file: eval-trace-conntrack.tex
\begin{figure}
  %%\vspace{-5mm}
  \centering
  \includegraphics[width=0.25\textwidth]{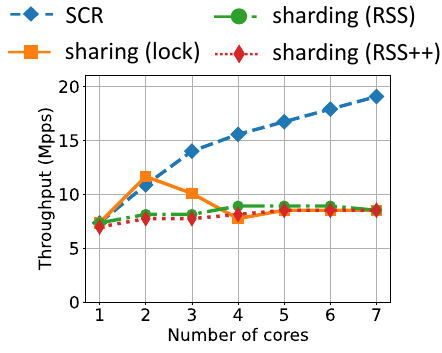}
  \vspace{-2mm}
  \caption{Throughput of TCP connection tracking
    parallelized using four techniques, \SCR (\Sec{design}), shared
    state, sharding with RSS, and sharding with
    RSS++~\cite{rss++-conext19}, on a hyperscalar data center trace
    (\Sec{eval-experiment-setup}). }
  %%\vspace{-6mm}
  \label{fig:throughput-conntrack-dctcp}
\end{figure}
%\vspace{-6mm}

%% file: eval-pcm-latency.tex
%\captionsetup[figure]{font={normalsize}}
\begin{figure}
  %\vspace{-5mm}
  \centering
  \includegraphics[width=0.48\textwidth]{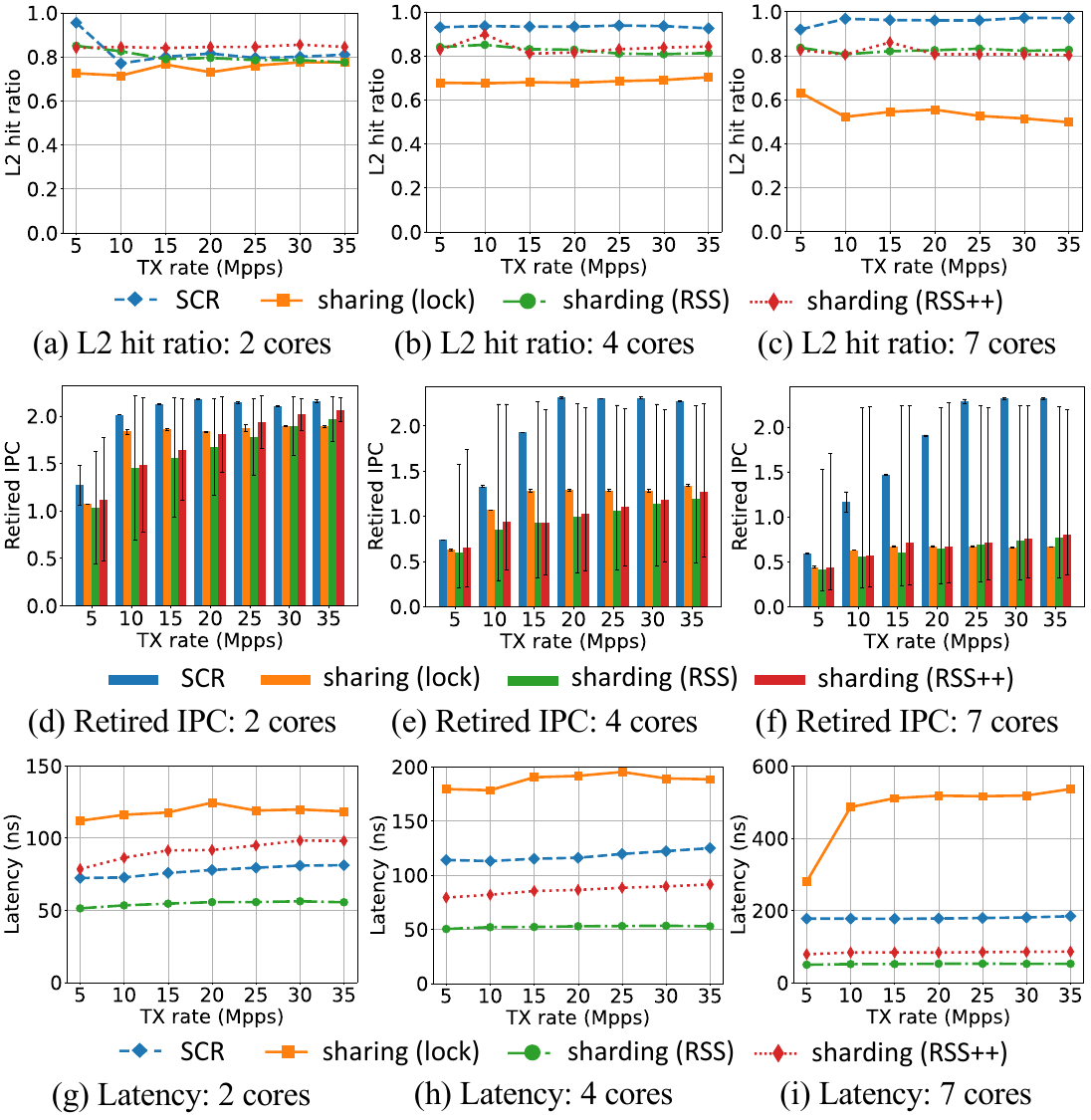}
  %%\vspace{-2.5mm}
  \caption{{Hardware performance metrics drawn from Intel PCM while
    executing the token bucket program. As the offered load
    increases, we show the program's compute latency (measured purely
    for the XDP portion), the L2 hit ratio, and the number of
    instructions retired per CPU clock cycle (IPC), when the program
    is scaled to 2, 4, or 7 cores. Packet traffic is from a
    university data center (\Sec{eval-experiment-setup}).}}
  \label{fig:pcm-explanations}
  \vspace{-5mm}
\end{figure}

%% file: eval-limits-scr.tex
\begin{figure}
  \vspace{-3mm}
  \centering
  \subfloat[Packets/second (1 rxq)]{
    \begin{minipage}[t]{0.155\textwidth}
    \includegraphics[width=\textwidth]{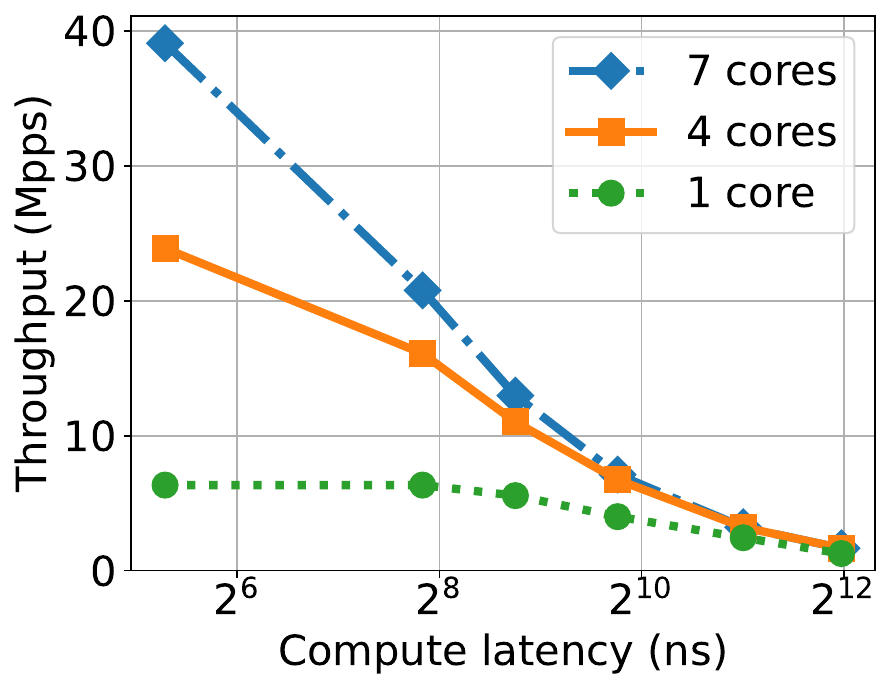}
    \label{fig:throughput-compute-time-increase-absolute-1rxq}
    \vspace{-4mm}
    \end{minipage}}
  \subfloat[Packets/second (2 rxq)]{
    \begin{minipage}[t]{0.155\textwidth}
    \includegraphics[width=\textwidth]{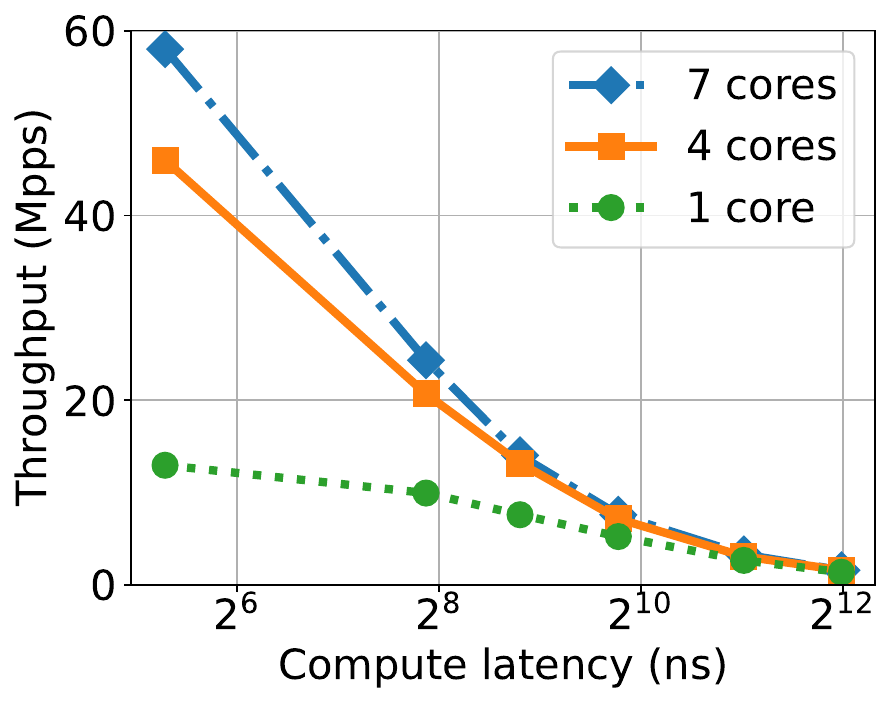}
    \label{fig:throughput-compute-time-increase-absolute-2rxq}
    \vspace{-4mm}
    \end{minipage}}
  \subfloat[Normalized to 1 core]{
    \begin{minipage}[t]{0.15\textwidth}
    \includegraphics[width=\textwidth]{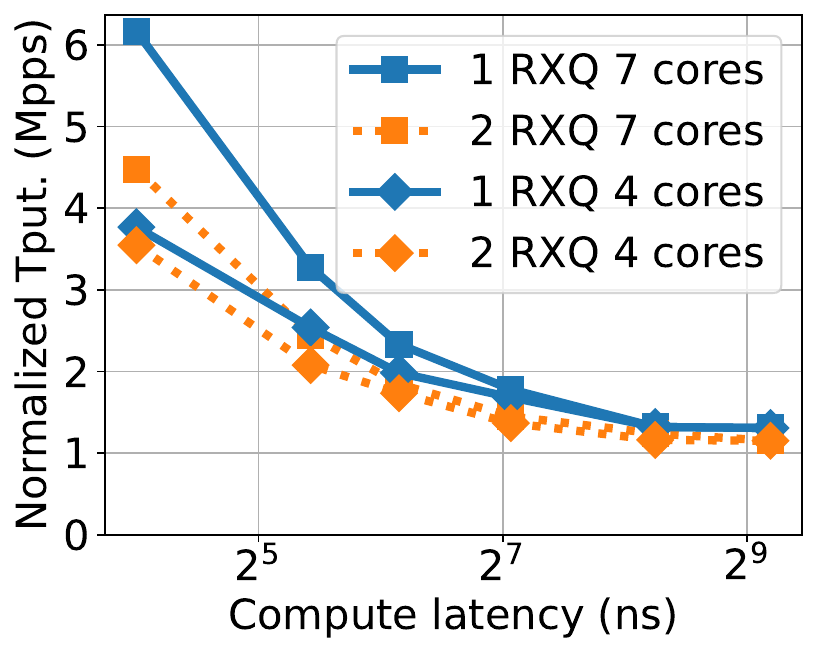}
    \label{fig:throughput-compute-time-increase-normalized}
    \vspace{-4mm}
  \end{minipage}}
  \vspace{-2.5mm}
  \caption{Evaluating the throughput scaling of a stateless program
    using \SCR, as the compute latency of the program varies
    but the dispatch latency remains constant, (a) in
    packets/second, and (b) normalized against single-core throughput
    at the same compute latency. As discussed in \Sec{scr-principle},
    the more the dispatch time dominates compute time, the more
    effective the multi-core scaling from \SCR. }
    \label{fig:compute-time-scr-limit}
    \vspace{-3mm}
\end{figure}

%% file: eval-loss-recovery.tex
\begin{figure}
  %%\vspace{2mm}
  \vspace{-2mm}
  \centering
  \includegraphics[width=0.48\textwidth]{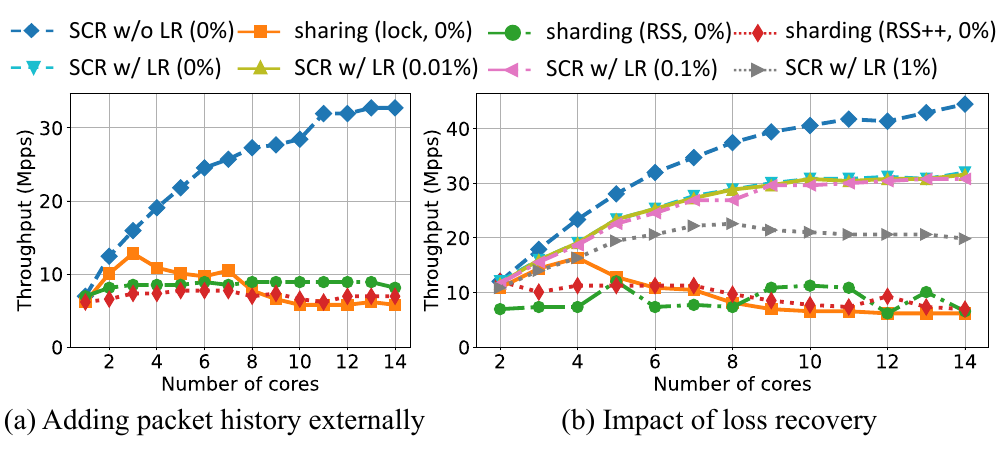}
  \vspace{-7mm}
  \caption{ (a) The throughput of a token bucket policer on the
    university data center trace (\Sec{eval-experiment-setup}), while
    truncating all packets in the trace to 64 bytes, with only \SCR
    adding metadata to packets before feeding them to the NIC.  (b)
    The throughput of a port-knocking firewall on the university data
    center trace. \SCR is run with and without loss recovery
    (\Sec{packet-loss-nondeterminism}) at multiple packet loss rates.}
    \label{fig:throughput-numbers-assorted}
    %%\vspace{-6mm}
\end{figure}

%% file: eval-hardware-synthesis-numbers.tex
\begin{center}
  \begin{table}
    \resizebox{0.38\textwidth}{!}{
      \begin{tabular}{|l|c|c|c|c|c|c|}
        \hline
            {\bf Rows} & \multicolumn{3}{|c|}{{\bf LUT}} & \multicolumn{2}{|c|}{{\bf Flip-flops}} \\
            & {\bf Usage} &  {\bf Logic} & {\bf \%} & {\bf Usage} & {\bf \%} \\
            \hline
            16  & 1045 &  646 & 0.060 & 2369 & 0.069 \\
            32  & 1852 & 1444 & 0.107 & 3158 & 0.091 \\
            64  & 2637 & 2229 & 0.153 & 4707 & 0.136 \\
            128 & 3390 & 2982 & 0.196 & 7786 & 0.226 \\
            \hline
      \end{tabular}
    }
    \caption{Sequencer resource usage after synthesis into the
      NetFPGA-PLUS reference switch and meeting timing at 250 MHz.}
    \label{tab:verilog-synthesis-results}
  \end{table}
\end{center}
\vspace{-6mm}

%% #entries        LUT       LUTRAM     FlipFlop       LUT(‰)         FlipFlop(‰)
%% 16                1045           408           2369             0.6                  0.69
%% 32                1852           408           3158             1.07                0.91
%% 64                2637           408           4707             1.53                1.36
%% 128              3390           408           7786             1.96                2.26
%% LUT minus LUTRAM equals to the LUTLOGIC resource usage.

%% file: eval-tofino-resource-numbers.tex
\begin{center}
  \begin{table}
    \resizebox{0.4\textwidth}{!}{
      %%\begin{tabular}{|l|c|c|c|c|c|c|}
      \begin{tabular}{|l|c|l|c|}
        \hline
            {\bf Resource} & {\bf Avg\%} & {\bf Resource} & {\bf Avg\%} \\
            \hline
            Exact match crossbars & 23.31\% & SRAM & 9.69\% \\
            VLIW instructions & 9.11\% & TCAM & 0.00\% \\
            Stateful ALUs & 93.75\% & Map RAM & 15.62\% \\
            Logical tables & 23.96\% & Gateway & 23.44\% \\
            \hline
      \end{tabular}
    }
    \caption{Resource usage (average \% across stages) of a Tofino
      implementation of the sequencer that uses as many stateful ALUs
      as possible to store packet history, amounting to 44 32-bit
      fields.}
    \label{tab:tofino-synthesis-results}
  \end{table}
\end{center}
\vspace{-14mm}

%% file: related.tex
\vspace{2mm}
\section{Related Work}

\forcameraready{High-performance packet processing is a deeply studied
research area. We covered the works most closely related to \SCR in
\Sec{motivation-background}. Here, we discuss other related work.
}
%
%% We cover prior work related to \SCR not addressed
%% in \Sec{motivation-background}.

\Para{Frameworks for network function performance.} The problem of
scaling out packet processing is prominent in network function
virtualization (NFV), with frameworks such as
split/merge~\cite{split-merge-nsdi13}, openNF~\cite{opennf-sigcomm14},
and Metron~\cite{metron-nsdi18} enabling elastic scaling.  There have
also been efforts to parallelize network functions
automatically~\cite{automatic-parallelization-nsdi24} and designing
data structures to minimize cross-core
contention~\cite{conntrack-hpsr21}. These efforts are flow-oriented,
managing and distributing state at flow granularity.  In
contrast, \SCR scales packet processing for a single
flow. 

\Para{General techniques for software parallelism.} Among the
canonical frameworks to implement software
parallelism~\cite{autoparallelization-lncs12}, our scaling principles
are most reminiscent of Single Program Multiple Data (SPMD)
parallelism, with the program being identical on each core but the
data being distinct. The sequencer in \SCR makes the data
distinct for each core.
\forcameraready{even when the inputs to each core are closely
related (\ie overlapping packet history sequences). In \SCR, the data
is made distinct for each core by the sequencer.}

\Para{Parallelizing finite state machines.} A natural model of
stateful packet processing programs is as finite state automata (the
state space is the set of flow states) making transitions on events
(packets). There have been significant efforts taken to parallelize
FSM execution using
speculation~\cite{fsm-parallelization-scalability-sensitive-sc17,
  fsm-parallelization-path-fusion-speculation-asplos21} and data
parallelism~\cite{data-parallel-fsm-asplos14}. In contrast, \SCR
exploits replication. %% which also provides the benefit of parallel
%% scale-out for the software \dispatch overheads, with a small added
%% cost of duplicated computation across cores.

\Para{Parallel network software stacks.} There has been recent
interest in abstractions and implementations that take advantage of
parallelism in network stacks, for TCP~\cite{flextoe-nsdi22,
tas-eurosys19} and for end-to-end data transfers to/from user
space~\cite{netchannel-sigcomm22}. \SCR takes a complementary
approach, using replication rather than decomposing the program into
smaller parallelizable computations.

%% file: conclusion.tex
\section{Conclusion}
\label{sec:conclusion}

It is now more crucial than ever to investigate techniques to scale
packet processing using multiple cores. This paper presented
state-compute replication (\SCR), a principle that enables scaling the
throughput of stateful packet-processing programs monotonically across
cores by leveraging a packet history sequencer, even under realistic
skewed packet traffic. 

\nop{
  It is now more crucial than ever to investigate techniques to scale
packet processing across multiple cores. This paper presented
state-compute replication (\SCR), a principle that enables scaling
stateful packet-processing programs across cores by leveraging a
packet history sequencer, even for a single stateful flow. \SCR
provides throughput benefits with linear scaling as more cores are
added, for realistic packet traffic, packet-processing benchmarks, and
system configurations.
}

\nop{
Given the saturation of single CPU core performance and the emerging
improvements in network line speeds, it is now more crucial than ever
to find ways to scale packet processing over multiple cores. A
significant barrier to achieving such scaling is the existence of
state, memory that is updated across multiple packets in a flow. This
paper presented state-compute replication (\SCR), a scaling principle
that enables stateful packet-processing programs to scale the
processing for a single flow across cores. Key to achieving such
scaling is a sequencer, a simple hardware component that can piggyback
recent packet history on each packet sprayed across cores. Far from
being niche, \SCR provides throughput benefits with linear scaling as
more cores are added, for realistic packet traffic, packet-processing
benchmarks, and system configurations. We hope that these principles
and its improvements enable much necessary innovation in high-speed
packet processing.
}

%% file: throughput-model.tex
\section{Throughput Model}
\label{app:throughput-model}
This section provides the detailed description of the 
model to predict the throughput of \SCR (outlined 
in \Sec{operationalizing-scr}) and evaluates whether the actual 
throughput matches the model.

Suppose a system has $k$ cores, and each core can dispatch a single packet
in $d$ cycles, run a packet-processing program that computes over
a single packet in $c = c_1 + (k-1) \times c_2$ cycles, where $c_1$
is the time for processing the current packet and $c_2$ is the time for
state transition using one metadata. $c_2$ is smaller than $c_1$, as 
state-transition code is a code snippet extracted from the program which 
processes the current packet. For each piggybacked packet, the
total processing time is $d + c_1 + (k-1) \times c_2$, where $t = d + c_1$.
When $t$ dominates state-computation time (\ie $t \gg c_2$), with $k$
cores, the total rate at which externally-arriving packets can be
processed is $k \times \frac{1}{t + (k-1) \times c} \approx k /
t$. \Tab{throughput-model-parameters} lists the parameters we measured
for packet-processing applications we evaluated.
It shows that $t$ is 4.3 -- 9.7 times $c_2$.
Hence, it is possible to scale the packet-processing rate linearly
with the number of cores $k$.

We applied the parameters in \Tab{throughput-model-parameters}
to the throughput model and compared the predicted throughput to the
the actual throughput. \Fig{throughput-model-mlffr} shows they match well.
\input{throughput-model-table-parameters}

\input{throughput-model-mlffr}

%% file: throughput-model-table-parameters.tex
\begin{table}
  % 65rsspp version
  \centering
  \begin{tabular}{|l|c|c|c|c|}
    \hline
    {\bf Application} & {\bf $t$} & {\bf $c_2$} & {\bf $d$} & {\bf $c_1$} \\
    \hline
    % DDoS mitigator & 142 & 9 & 115 & 27 \\ % 7 cores
    DDoS mitigator & 114 & 15 & 104 & 10 \\ % 14 cores
    Heavy hitter monitor & 145 & 15 & 110 & 35 \\
    Token bucket policer & 156 & 21 & 104 & 53 \\
    % Port-knocking firewall & 139 & 12 & 111 & 28 \\ % 7 cores
    Port-knocking firewall & 107 & 18 & 97 & 11 \\ % 14 cores
    TCP connection tracking & 152 & 35 & 80 & 73 \\  
    \hline
  \end{tabular}
  \caption{The throughput model parameters (in nanoseconds) for packet-processing applications we evaluated.}
  \label{tab:throughput-model-parameters}
\end{table}

% \begin{table}
%   \centering
%   \begin{tabular}{|l|c|c|c|c|}
%     \hline
%     {\bf Application} & {\bf $t$} & {\bf $c_2$} & {\bf $d$} & {\bf $c_1$} \\
%     \hline
%     DDoS mitigator & 114 & 10 & 91 & 24 \\
%     Heavy hitter monitor & 86 & 16 & 55 & 31 \\
%     Token bucket policer & 90 & 21 & 39 & 51 \\
%     Port-knocking firewall & 106 & 12 & 79 & 27 \\
%     TCP connection tracking & 132 & 34 & 60 & 73 \\  
%     \hline
%   \end{tabular}
%   \caption{The throughput model parameters (in nanoseconds) for packet-processing applications we evaluated.}
%   \label{tab:throughput-model-parameters}
% \end{table}

%% file: throughput-model-mlffr.tex
\begin{figure}
  \vspace{-5mm}
  \centering
  \subfloat[DDoS mitigation]{
    \begin{minipage}[t]{0.16\textwidth}
    \includegraphics[width=\textwidth]{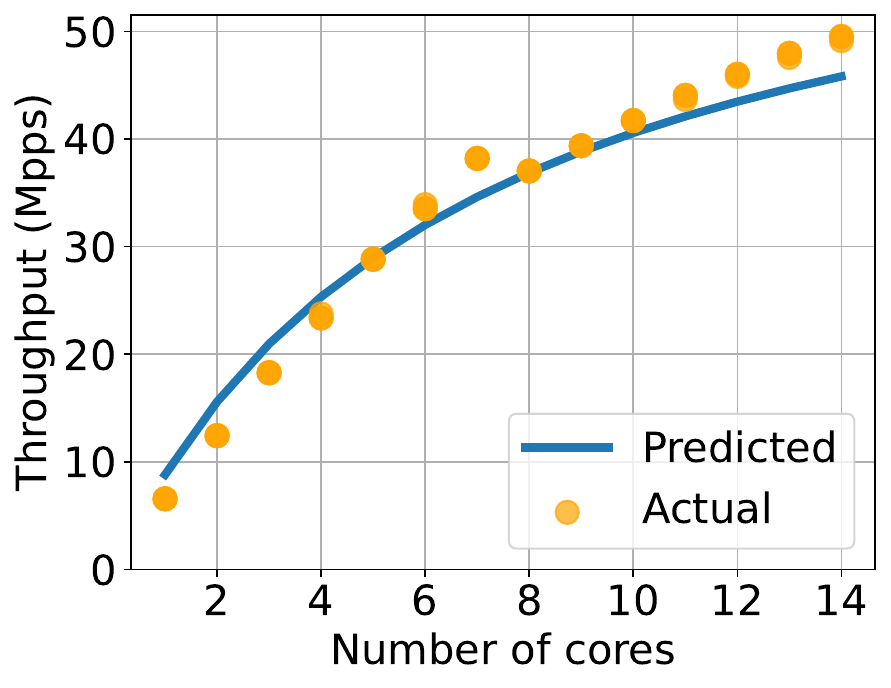}
    \label{fig:throughput-model-ddos-dc}
    \vspace{-4mm}
    \end{minipage}} 
    % \hspace{0.01\textwidth}%
  \subfloat[Heavy hitter detector]{
    \begin{minipage}[t]{0.16\textwidth}
    \includegraphics[width=\textwidth]{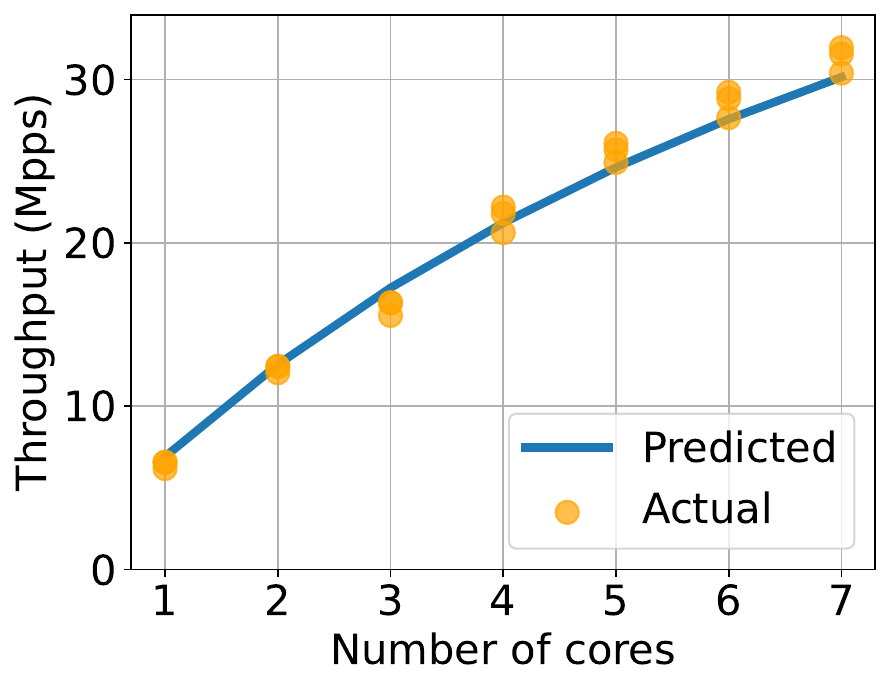}
    \label{fig:throughput-model-heavyhitter-dc}
    \vspace{-4mm}
    \end{minipage}}
  \subfloat[Token bucket policer]{
    \begin{minipage}[t]{0.16\textwidth}
    \includegraphics[width=\textwidth]{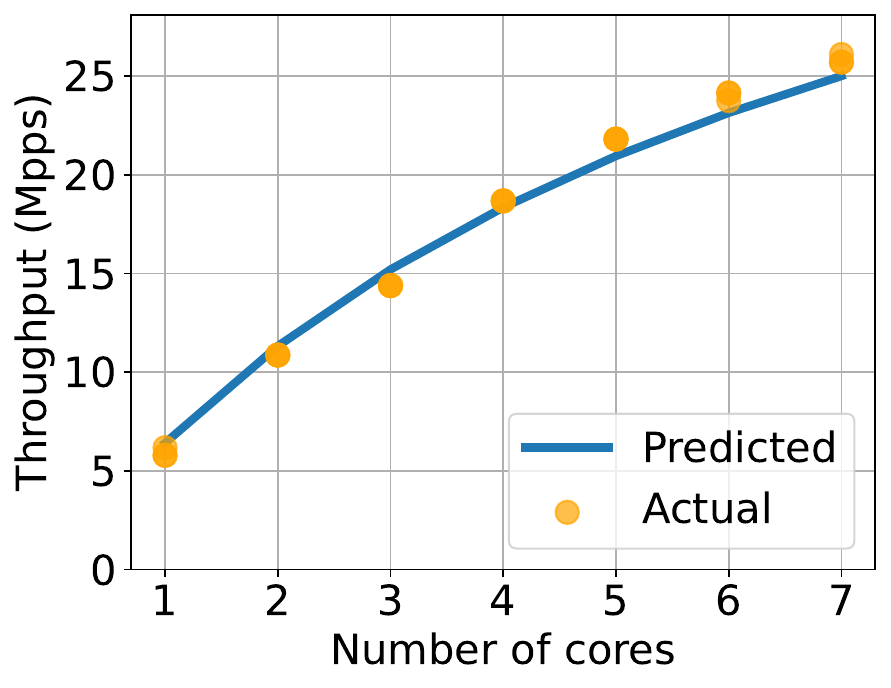}
    \label{fig:throughput-model-tokenbucket-dc}
    \vspace{-4mm}    
    \end{minipage}} \\
    % \hspace{0.1\textwidth}%
  \subfloat[Port-knocking firewall]{
    \begin{minipage}[t]{0.16\textwidth}
    \includegraphics[width=\textwidth]{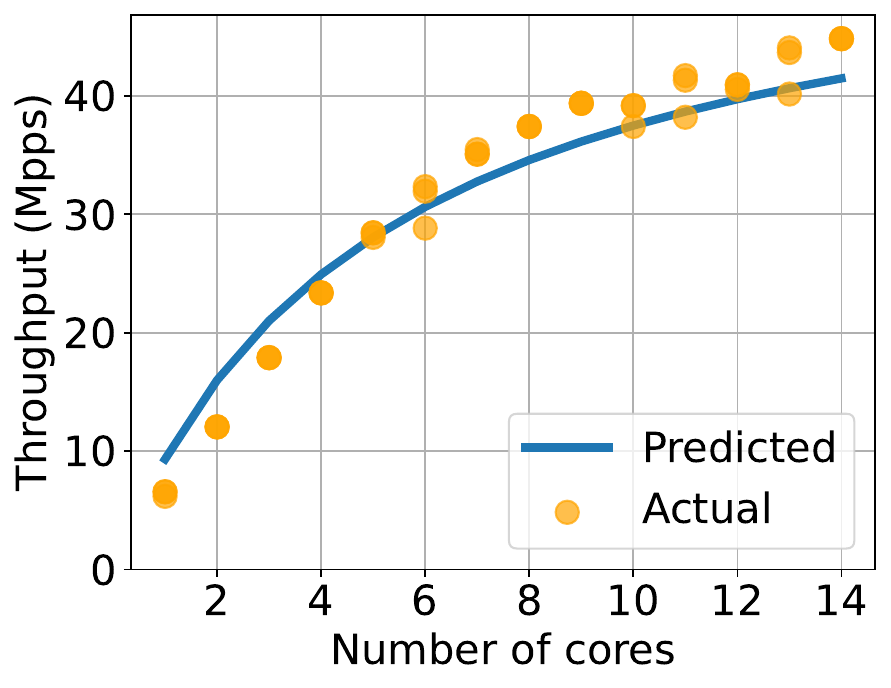}
    \label{fig:throughput-model-portknocking-dc}
    \vspace{-4mm}
    \end{minipage}}
  \subfloat[TCP connection tracking]{
    \begin{minipage}[t]{0.18\textwidth}
    \includegraphics[width=0.9\textwidth]{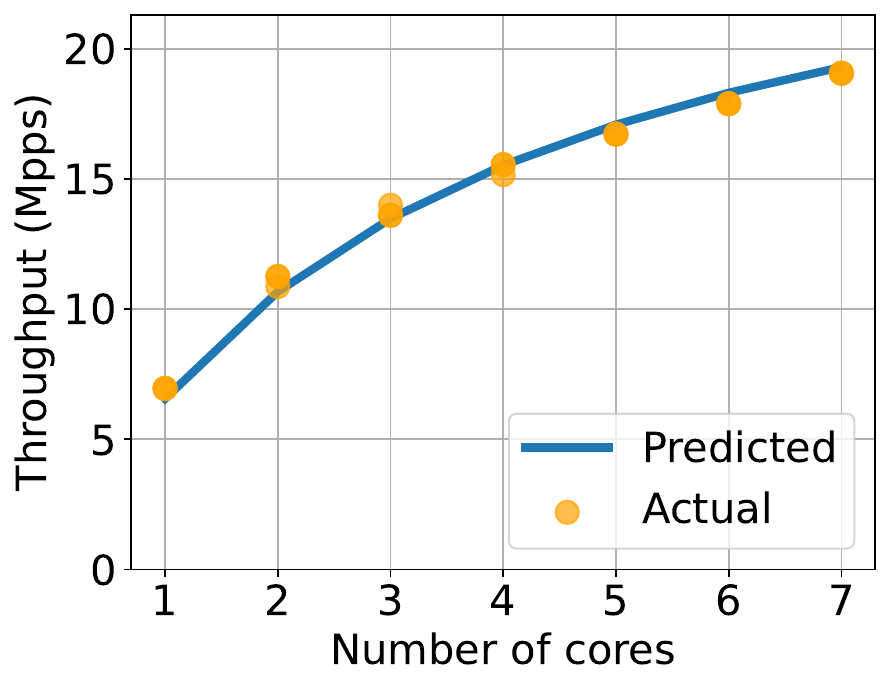}
    \label{fig:throughput-model-conntrack-dc}
    \vspace{-4mm}
    \end{minipage}}
  \caption{Predicted and actual throughput (\Sec{eval-experiment-setup})
    in millions of packets per second (Mpps) of five stateful packet-processing
    programs implemented using \SCR (\Sec{design}). The workloads of 
    (a)-(d) and (e) are from a university data center and a hyperscalar 
    data center (\Sec{eval-experiment-setup}) separately.}
  \label{fig:throughput-model-mlffr}
  \vspace{-5mm}
\end{figure}

%% file: loss-recovery.tex
\section{Loss Recovery Algorithm}
\label{app:loss-recovery-algorithm}

% \Alg{detect_lost_pkts} is the pseudo code of detecting lost packets. 
% For core $c$, we use a per-core static variable $max\_last\_pkt_c$ to keep track of the maximum packet number received by $c$, $maxp(sp)$. Upon receiving an $sp$, we detect packet loss by checking if ``packet number continuity" (line 2-3) holds, and update the lost packet sequence $lost\_pkts$ (line 4-6). Finally, we update $max\_last\_pkt_c$ to the maximum packet number in the current \SCR packet.
\renewcommand{\algorithmicrequire}{\textbf{Input:}}
\renewcommand{\algorithmicensure}{\textbf{Output:}}
\algnewcommand{\Initialize}[1]{%
  \State \textbf{Initialize:}
  \Statex \hspace*{\algorithmicindent}\parbox[t]{.8\linewidth}{\raggedright #1}
}
\algnewcommand\algorithmicforeach{\textbf{for each}}
\algdef{S}[FOR]{ForEach}[1]{\algorithmicforeach\ #1\ \algorithmicdo}
\algnewcommand\Continue{\textbf{continue}}
\algnewcommand\Break{\textbf{break}}
\begin{algorithm}
\caption{\SCR loss recovery}\label{alg:loss_recovery}
\begin{algorithmic}[1]
\Require the received \SCR pkt ($sp$), current core ($c$)
\Initialize{\strut$max[c] \gets 0$, $\forall i: log[c][i] \gets NOT\_INIT$}
\Function{$scr\_loss\_recovery$}{$sp$, $c$}
% \Ensure per-core log $log$
\State $maxseq \gets sp.seq\_num $
\State $minseq \gets max(1, maxseq-N+1) $
% \State $lost_{min}, lost_{max} \gets detect(maxseq, max[c])$
\For{$i \gets max[c] + 1$ to $maxseq$}
    % \State $info.lost$ $\gets$ true
    \If{$i < minseq$}
        \State $log[c][i] \gets LOST$
        \State $handle\_loss\_recovery(i, c)$
    \Else
        \State $history$ $\gets$ get history of $i$ from $sp$
        \State $log[c][i] \gets history$
    \EndIf
\EndFor

\State $max[c] \gets maxseq$

\EndFunction
\Statex
\Function{$handle\_loss\_recovery$}{$i$, $c$}
\State $C_{others} \gets C \setminus \{c\} $
% \color{blue}
\State $C_{lost} \gets \emptyset $
% \color{black}
\While{true}
    \ForEach {$c' \in \mathcal C_{others} $}
        \If{$log[c'][i]$ is $NOT\_INIT$}
            \State continue
        % \color{blue}
        \ElsIf{$log[c'][i]$ is $LOST$}
            \State $C_{lost} \gets$ $C_{lost} \cup \{c'\} $
            \If{$C_{lost} = C_{others}$}
                \Return
            \EndIf
         % \color{black}
        \Else
            \State $history \gets log[c'][i]$ 
            \State update state using $history$
            \State \Return
        \EndIf
    \EndFor
\EndWhile
\EndFunction
\end{algorithmic}
\end{algorithm}

% \begin{algorithm}
% \caption{$detect$}\label{alg:detect_lost_pkts}
% \begin{algorithmic}[1]
% \Require Max. seq \# in $sp$ ($maxseq$),
% Max. seq \# core $c$ has seen before $sp$ ($max[c]$)
% \Ensure lost regular pkt $\#$ range $lost_{min}$, $lost_{max}$
% \State $minseq \gets max(1, maxseq-N +1)$
% \State $lost \gets (max[c] + 1 < minseq)$
% \State $lost_{min} \gets 0, lost_{max} \gets 0$
% \If{lost = true}
%     \State $lost_{min} \gets max[c] + 1$
%     \State $lost_{max} \gets minseq - 1$
% \EndIf
% \State \Return $lost_{min}, lost_{max}$
% \end{algorithmic}
% \end{algorithm}

\newtheorem{lemma}{Lemma}
\newtheorem{theorem}{Theorem}
This section provides the detailed pseudocode and a proof 
of correctness of the packet loss recovery 
algorithm outlined in \Sec{packet-loss-nondeterminism}.

Before we start the correctness proof of loss recovery algorithm,
we define the following notations.
\begin{CompactEnumerate}
    \item $sp$: an \SCR packet. $sp_i$: the $i^{th}$ \SCR packet sent 
    from the sequencer to a core (please refer to \Sec{packet-format} 
    for more details).
    \item $p$: a regular packet. $p_i$: the $i^{th}$ regular packet 
    received by the sequencer (in original or as part of the packet
    history, please refer to \Sec{packet-format} for more details).
    \item $C$: the collection of all cores.
\end{CompactEnumerate}

We want to prove every core will not be deadlocked by loss recovery, 
\ie every core will \emph{start processing} (execute line 6)
and \emph{finish processing} (finish executing line 6-12) every regular packet,
under the conditions (1) each core will receive at least one \SCR packet 
after packet loss, (2) we have infinite memory, and (3) packet number 
monotonically increases.

\begin{theorem} \label{theorem:loss_recovery}
For any regular packet $p_i$, if every core has received $p_j$ ($j \ge i$), 
every core will finish processing $p_i$.
\end{theorem}

\begin{proof}
We will firstly prove any core will process $p_1$ to $p_i$ in order.
The order of packets to process follows the order at line 5 
(from $max[c]+1$ to $maxseq$), \ie, after a core finishing 
processing $p_k$, it will start processing $p_{k+1}$.
Each core will be triggered to process all packets including $p_1$ to $p_i$,
since each core has received $p_j$ ($j \ge i$),

Given the order of packets to process, we now prove all cores can finish
processing $p_i$.
If $i=1$, all cores will start processing $p_1$ (line 6) and then finish 
processing $p_1$ (Lemma \ref{lemma:induction}).
If $i>1$, according to the order of a single core processing 
packets and Lemma \ref{lemma:induction}, we get the induction hypothesis 
that if all cores have started processing $p_k$, all cores will 
finish processing $p_{k}$ and start processing $p_{k+1}$.
Using induction hypothesis for $i$ times, all cores will 
start and then finish processing $p_i$.
\end{proof}

\begin{lemma} \label{lemma:induction}
For any regular packet $p_i$, if all cores have started processing $p_i$, 
then all cores will finish processing $p_i$.
\end{lemma}

\begin{proof}
For any core $c$, no matter $p_i$ is lost or received,
$c$ will finish processing it. 

If $p_i$ is received by $c$, after $c$ updates $history[i]$ 
in its log (line 10-11), $c$ finishes processing $p_i$. 

If $p_i$ is lost at $c$ (detected at line 6), $c$ will wait for other cores to
update $p_i$ in their logs until $c$ gets $history[i]$ or confirm $p_i$ is 
lost at all of other cores (line 19-33). $c$ will not be deadlocked in waiting, 
since $p_i$ will be updated to $history[i]$ or $LOST$ in the logs of 
all cores who have started processing $p_i$ (if line 6 is executed for $p_i$, 
line 7 or line 11 will be executed).
\end{proof}

Note that logs are finite and sequence 
numbers wrap around in real system, but these can be handled with a 
sufficiently large log and sequence space, and we use values 1,024 and 
842,185 in our current implementation.

%% file: running-example-programming.tex
\section{\SCR-Aware Multi-Core Programming}
\label{sec:scr-programming}
\label{app:scr-programming}

\begin{wrapfigure}{r}{0.25\textwidth}
  %%\vspace{-5mm}
  \centering
  \includegraphics[width=0.24\textwidth]{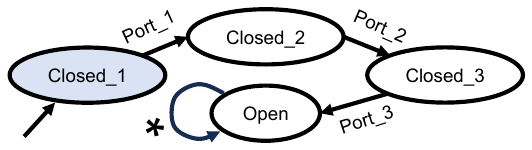}
  \vspace{-2mm}
  \caption{A state machine for a simple port-knocking firewall.}
  \vspace{-3mm}
  \label{fig:port-knocking-state-machine}
\end{wrapfigure}

Consider a packet-processing program developed assuming
single-threaded execution on a single CPU core. The question we tackle
in this subsection is: how should the program be changed to take
advantage of multi-core scaling with state-compute replication? We
walk through the process of adapting a program written in the eBPF/XDP
framework~\cite{xdp-conext18}, but we believe it is conceptually
similar to adapt programs written in other frameworks such as DPDK.

We describe the program transformations necessary for \SCR through a running
example. Suppose we have a port-knocking
firewall~\cite{port-knocking-firewall} with the state machine shown in
\Fig{port-knocking-state-machine}. The program runs a copy of this
state machine per source IP address. If a source transmits IPv4/TCP
packets with the correct sequence of TCP destination ports, then all
further communication is permitted from that source. All other packets
are dropped. Any transition not shown in the figure leads to the
default {\ct CLOSED\_1} state, and only the {\ct OPEN} state permits
packets to traverse the firewall successfully. A simplified XDP
implementation of this single-threaded firewall is shown below.

\input{xdp-portknock-full-example}

The program's state is a key-value dictionary mapping source IP
addresses to an automaton state described in
\Fig{port-knocking-state-machine}. The function {\ct get\_new\_state}
implements the state transitions. The main function, {\ct
  simple\_port\_knocking} first parses the input packet, dropping
all packets other than IPv4/TCP packets. Then the program fetches the
recorded state corresponding to the source IP on the packet, and
performs the state transition corresponding to the TCP destination
port. If the final state is {\ct OPEN}, all subsequent packets of that
source IP may traverse the firewall to the other side. All other
packets are dropped.

To enable this program to use state-compute replication across cores,
this program should be transformed in the following ways. We believe
that these transformations may be automated by developing suitable
compiler passes, but we have not yet developed such a compiler.

\noindent (1) {\em Define per-core state data structures and
  per-packet metadata structures.} First, the program's state must be
replicated across cores. To achieve this, we must define per-core
state data structures that are identical to the global state data
structures, except that they are not shared among CPU
cores. Packet-processing frameworks provide APIs to define such
per-core data structures~\cite{bpf-maps}.

Additionally, we must define a per-packet metadata structure that
includes any part of the packet that is used by the program---through
either control or data flow---to update the state corresponding to
that packet. For the port-knocking firewall, the per-packet metadata
should include the {\ct l3proto, l4proto, srcip,} and {\ct dport}.

The data structures that maintain packet history on the sequencer
correspond to this per-packet metadata (\Sec{pipeline-sequencer}).

\noindent {\em (2) Fast-forward the state machine using the packet
  history.} The \SCR-aware program must prepend a loop to ``catch up''
the state machine for each packet missed by the CPU core where the
current packet is being processed. By leveraging the recent history
piggybacked on each packet, at the end of this loop, the CPU core has
the most up-to-the-packet state.

\input{xdp-portknock-fast-forwarding}

%% \ngs{simplify the `for' loop structure using a different counter
%%   variable}

A few salient points about the code fragment above. First, the
semantics of the ring buffer of packet history
(\Sec{pipeline-sequencer}) are implemented by looping over the packet
history metadata starting at offset {\ct index} rather than at offset
0. The decision to implement the ring buffer semantics in software
makes the hardware significantly easier to design, since only a small
part of the hardware data structure needs to be updated for each
packet (\Sec{data-structure-packet-history}). Second, the loop must
implement appropriate control flow before the state update to ensure
that only packets that should indeed update the flow state do. Note
that the metadata includes parts of the packet that are not only the
data dependencies for the state transition ({\ct srcip, dport}) but
also the control dependencies ({\ct l3proto, l4proto}). Third, no
packet verdicts are given out for packets in the history: we want the
program to return a judgment for the ``current'' packet, not the
historic packets used merely to fast-forward the state
machines. Finally, the code fragment conveniently adjusts {\ct
  pkt\_start} to the position in the packet buffer
(\Fig{packet-format}) corresponding to where the ``original'' packet
begins. The rest of the original program---unmodified---may process
this packet to completion and assign a verdict.

What is excluded in our code transformations is also crucial.  This
program avoids locking and explicit synchronization, despite the fact
that it runs on many cores, even if there is global state maintained
across all packets.

With these transformations, in principle, a packet-processing program
is able to scale its performance using state-compute replication
across multiple cores.

%% file: xdp-portknock-full-example.tex
\begin{lstlisting}

/* Definition of program state */
struct map states {
  /* assume we define a dictionary with keys
     as source IP addresses and values as firewall
     states among CLOSED_{1, 2, 3} and OPEN. */
}

/* State transition function. See (*\Fig{port-knocking-state-machine}*) */
int get_new_state(int curr_state, int dport) {
  /* A function that implements the state machine for
     the port knocking firewall. */
  if (curr_state == CLOSED_1 && dport == PORT_1)
    return CLOSED_2;
  if (curr_state == CLOSED_2 && dport == PORT_2)
    return CLOSED_3;
  if (curr_state == CLOSED_3 && dport == PORT_3)
    return OPEN;
  if (curr_state == OPEN)
    return OPEN;
  return CLOSED_1;
}

/* The main function */
int simple_port_knocking(...) {
  /* Assume the packet is laid out as a byte array
     starting at the address pkt_start. Suppose the
     packet is long enough to include headers up to
     layer 4. First, parse IPv4/TCP pkts. */
  struct ethhdr* eth = pkt_start; // parse Ethernet
  int l3proto = eth->proto; // layer-3 protocol
  int off = sizeof(struct ethhdr);   
  struct iphdr* iph = pkt_start + off;
  int l4proto = iph->protocol; // layer-4 protocol
  if (l3proto != IPv4 || l4proto != TCP)
      return XDP_DROP; // drop non IPv4/TCP pkts
  int srcip = iph->src; // source IP addr 
  off += sizeof(struct iphdr);
  struct tcphdr* tcp = pkt_start + off;
  int dport = tcp->dport; // TCP dst port
  
  /* Extract & update firewall state for this src. */
  int state = map_lookup(states, srcip);
  int new_state = get_new_state(state, dport);
  map_update(states, srcip, new_state);

  /* Final packet verdict */
  if (new_state == OPEN)
      return XDP_TX; // allow traversal
  return XDP_DROP; // drop everything else
}
\end{lstlisting}

%% file: xdp-portknock-fast-forwarding.tex
\begin{lstlisting}
/* Assume the pointer `data' locates where the
   per-pkt metadata begins in the byte array of
   the packet ((*\Fig{packet-format}*)). Suppose `index'
   is the offset of the earliest packet (*\Sec{data-structure-packet-history}*),
   and NUM_META is the number of packets in the
   piggybacked history.
*/   
int l3proto, l4proto, srcip, dport, i, j;
for (j = 0; j < NUM_META; j++) {
  i = (index + j) % NUM_META; // ring buffer
  struct meta *pkt = data + i * sizeof(meta);
  l3proto = pkt->l3proto;
  l4proto = pkt->l4proto;
  srcip   = pkt->srcip;
  dport   = pkt->dport;
  if (l3proto != IPv4 || l4proto != TCP)
     continue; // no state txns or pkt verdicts
  /* Update state for this srcip and dport: */
  /* map_lookup; get_new_state; map_update. */
  /* Note: No pkt verdicts for historic pkts. */
}
pkt_start = data + NUM_META * sizeof(struct meta)
              + sizeof(index);
\end{lstlisting}